\documentclass{article}

\usepackage{arxiv}

\usepackage[utf8]{inputenc} 
\usepackage[T1]{fontenc}    
\usepackage{fancyhdr}
\usepackage{lmodern}
\usepackage[colorlinks]{hyperref}  
\usepackage{url}            
\usepackage{booktabs}       
\usepackage{amsfonts}       
\usepackage{nicefrac}       
\usepackage{microtype}      
\usepackage{lipsum}

\usepackage{graphicx}
\usepackage{epsfig}
\usepackage{epstopdf}
\usepackage{subfigure}

\usepackage{enumerate}
\usepackage{amssymb}
\usepackage{amsmath}
\usepackage{amsthm}
\usepackage{enumitem}
\usepackage{mathrsfs}
\usepackage{rotating}
\usepackage{bm}

\newcommand{\R}{\mathbb{R}}

\newcommand{\dis}{\displaystyle}

\renewcommand{\bar}{\overline}
\renewcommand{\rho}{\varrho}
\renewcommand{\phi}{\varphi}

\newtheorem{thm}{Theorem}[section]
\newtheorem{rem}{Remark}[section]

\title{A reaction-diffusion model for Mycobacterium tuberculosis infection}

\author{  
C. Accarino$^1$\href{https://orcid.org/0000-0003-2345-9423}
{\includegraphics[scale=0.06]{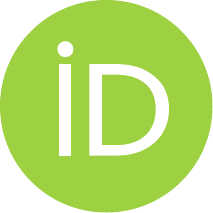}} 
\\ 
\texttt{chiara.accarino@unina.it} \\
\And
R. Accarino$^2$ 
\\ 
\texttt{rossella.accarino@unina.it} \\
\And
F. Capone$^1$\thanks{Corresponding author.} \href{https://orcid.org/0000-0002-0672-999X}{\includegraphics[scale=0.06]{orcid.eps}} \\ 
\texttt{fcapone@unina.it} \\
\And  
R. De Luca$^1$\href{https://orcid.org/0000-0002-2109-7564}{\includegraphics[scale=0.06]{orcid.eps}} \\ 
\texttt{roberta.deluca@unina.it} \\
\And
L. Fiorentino$^1$\href{https://orcid.org/0000-0002-0154-6035}{\includegraphics[scale=0.06]{orcid.eps}} \\ \texttt{ludovica.fiorentino@unina.it} \\ 
\And
G. Massa$^1$\href{https://orcid.org/0000-0002-8401-9176}{\includegraphics[scale=0.06]{orcid.eps}} \\ \texttt{giuliana.massa@unina.it}  
\\ \\ $^1$Dipartimento di Matematica e Applicazioni 'R.Caccioppoli' \\ Università degli Studi di Napoli Federico II \\ Via Cintia, Monte S.Angelo, 80126 Napoli \\ Italy \\
\\ \\ $^2$Pathology Unit, Department of Advanced Biomedical Sciences, \\ Università degli Studi di Napoli Federico II \\ Napoli \\ Italy \\ }

\renewcommand{\shorttitle}{A reaction-diffusion model for Mycobacterium tuberculosis infection}

\begin{document}
\maketitle
\begin{abstract}
This paper aims to investigate a reaction-diffusion model which describes in-host infection for \emph{Mycobacterium tuberculosis} (Mtb) allowing random motion (i.e. linear diffusion) and chemotaxis of macrophages and bacteria populations. In particular, chemotaxis-driven aggregation of macrophages plays a fundamental role in the development of the Mtb infection and the production of chemokine -- located in the infection site -- represents an attractant for the uninfected macrophages, therefore we consider chemotaxis between infected macrophages and uninfected macrophages. The linear stability of the endemic equilibria is investigated. In particular, by looking for conditions guaranteeing that an equilibrium, stable in the absence of diffusion, becomes unstable when diffusion is allowed, the formation of Turing patterns -- that biologically represent the formation of granuloma initiated by the immune cells -- is investigated. Furthermore, a weakly nonlinear analysis is performed to deeply explore the patterns amplitude. 
\end{abstract}

\section{Introduction}
Tuberculosis (TB), also known as the \emph{white death}, is an infectious disease caused by the bacillus \emph{Mycobacterium tuberculosis} (Mtb), and, according to the World Health Organization (WHO), in 2022 it was the second leading infectious killer after COVID-19 \cite{WHO2023}. \\
TB is initially transmitted via the respiratory route when infected people expel the bacteria into the air by coughing, sneezing or spitting, and can subsequently spread throughout the rest of the body, reaching the meninges, the lymph nodes, the bones, the stations of the monocyte-macrophage system, the abdominal organs and the urogenital system. Once \emph{Mycobacterium tuberculosis} is inhaled into the lungs, it is phagocytosed by alveolar macrophages, which are the cells of the innate immunity that primarily intervene to recognise the bacterium and to activate the cells of the adaptive immunity. \emph{Mycobacterium tuberculosis} prevents the fusion of the phagosome with the lysosomes, meanwhile the phagosome is able to fuse with other intracellular vesicles, ensuring access to nutrients and facilitating intravacuolar replication. The alveolar macrophages invade the underlying epithelial layer, and induce the activation of the immune system which sends the T-lymphocytes to the site of infection via blood vessels. These immune cells start a chronic inflammatory reaction called \emph{granuloma} which limits the spread of the bacteria, that are incorporated into the macrophages, by isolating them from the rest of the lung. Precisely, a tubercular granuloma is composed by a central area of caseous necrosis surrounded by epithelioid macrophages, lymphocytes and multinucleate giant cells. In Fig. \ref{fig1} is shown a picture of a slide of a caseating granuloma in a male patient of 23 years who was diagnosed the TBC in 2023 by the Pathology Unit of the Department of Advanced Biomedical Sciences of the University of Naples "Federico II". 
\begin{figure}
    \centering
    \includegraphics[scale=0.25]{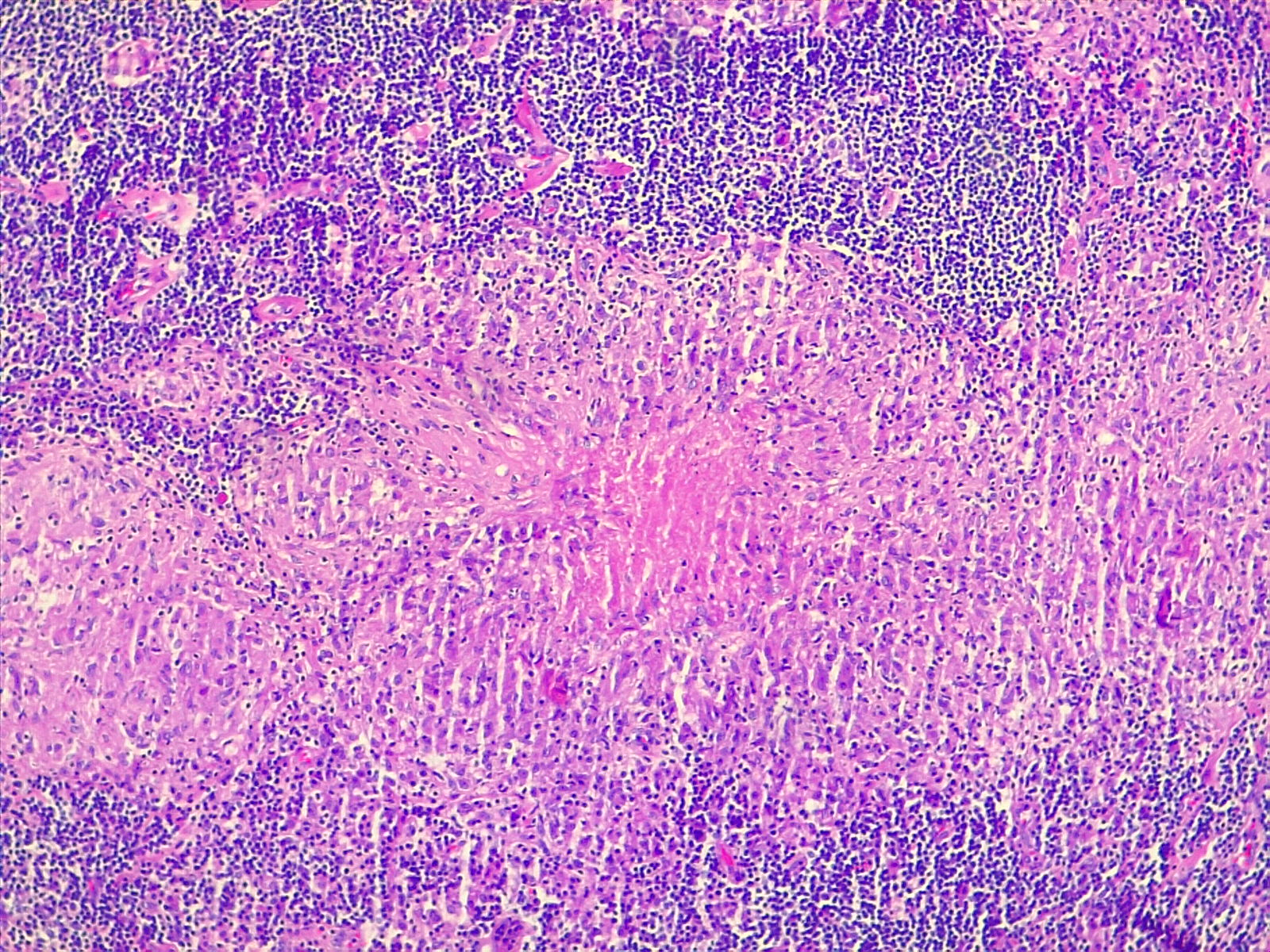}
    \caption{Caseating granuloma (haematoxylin and eosin, original magnification, x10).}
    \label{fig1}
\end{figure}
After the contagion, the immune system can eradicate the bacterium, so the disease is not developed, and the so-called clearance occurs. Alternatively, there may be a period of latency during which the disease is incubated, and in this period it can either happen that the bacterium is eradicated -- and also in this case we talk about clearance -- or that it develops the primary tuberculosis. Finally, it is possible to recover from the disease in this case too, however there is a risk that the bacterium is not definitively eradicated and that the disease can therefore be reactivated over the years, especially in the case of immune deficiency. Generally, after initial infection, the incubation period lasts from two to ten weeks, but the risk of developing the active disease is highest in the first two years. However, most infected people don't develop active disease, namely they remain with latent infection. The symptoms may be mild for many months, so it is easy to spread TB to others without knowing it. The common symptoms are prolonged cough (sometimes with blood), chest pain, weakness, fatigue, weight loss, fever and night sweats. Moreover, certain conditions (diabetes, weakened immune system, being malnourished and tobacco use) can increase the risk of developing active TB.
Without treatment, the death rate caused by TB disease is about $50\%$, while with $4-6$ months of anti-TB drug treatments about $85\%$ of infected people can be cured. Actually, according to WHO, in 2022 about 1.3 million people died. An estimated 10.6 million people fell ill with tuberculosis (TB) worldwide, including 5.8 million men, 3.5 million women and 1.3 million children \cite{WHO2023}.\\ 
\noindent
To study \emph{Mycobacterium tuberculosis} infection, mathematical models play a central role. Although in these models only some variables describing the phenomenon are considered, they are able to give some estimates about the spread of epidemics. In this way, the asymptotic behaviour of infection can be predicted and it is possible to try to control epidemics. Mathematical models for the transmission of \emph{Mycobacterium tuberculosis} have been widely studied in the literature (see for example \cite[and the reference therein]{Du2017,ZFH2020,ZFH2021,zhang2020,natalini2010}). In \cite{Du2017}, they considered a four-dimensional ODE system describing an in-host infection model, which they reduced to a three-dimensional model. The occurrence of periodic solutions, the existence, under certain conditions and by considering infected macrophages killing rate constant, of backward bifurcation, and the global stability of the system have been demonstrated. In \cite{ZFH2020,zhang2020} the authors have considered the complete four-dimensional ODE model, and they've demonstrated that there exist four different coexistence equilibria which describe the four different disease outcomes (primary tuberculosis, latent tuberculosis, clearance and reactivation), and they also found the conditions for forward and backward bifurcations. In \cite{ZFH2021}, the crucial factors describing various disease outcomes are determined.\\
\noindent
In order to better describe the phenomenon of the infection caused by Mtb, it is necessary to take into account the diffusion of the considered species. The reaction-diffusion models and their applications have been widely studied in literature. These models are very useful to describe phenomena like morphogenetic processes \cite{T1952}, medicine \cite{GLS2021}, population dynamics \cite{CDLFLM2023}, epidemiological systems \cite{CDL2017} and chemical reactions \cite{CDLT2018}, \cite{GLSS2013}. \\
\noindent
It has been proved in literature the random motion (self-diffusion) and directed cell-movement (chemotaxis) of macrophages and bacteria around the site of infection \cite[and references therein]{GammackEtAl2004,symon1972plasma}, therefore in this paper we generalize the mathematical model proposed by \cite{Du2017}, and further studied by \cite{ZFH2020,ZFH2021,zhang2020}, that describes in-host infection for Mycobacterium tuberculosis (Mtb), taking into account spatial non-homogeneously distribution of three heterogeneity distributed populations (uninfected macrophages -- infected macrophages -- bacteria), allowing self-diffusion and chemotaxis of the considered populations, in order to describe the \textit{formation of granuloma}, a pattern that, from a mathematical point of view, can be obtained via the analysis of an equilibrium state between bacteria and macrophages \cite{natalini2010}. Generally, geometrical patterns in bacterial populations can occur from local interactions between the considered cellular populations, depending on the interplay between these populations. In the presence of certain chemicals, many bacteria tend to move due to a gradient of concentration. If they move toward higher concentration of chemical, then it is a chemoattractant, while if they move toward a lower concentration of chemical, then it is a repellent \cite{Murray2_2002}. During Mtb infection, the immune system cells migrate to the infection site due to the release of chemokines by epithelial cells, macrophages, T-cells and bacteria. In particular, infected macrophages release various chemokines (e.g. IL-8, MIP2, IP-10, MCP-1) that attract uninfected macrophages, neutrophils and T cells to sites of infection \cite[and references therein]{natalini2010,GammackEtAl2004,symon1972plasma}. \\
When studying Turing pattern formation, weakly nonlinear analysis serves as a fundamental tool for comprehending pattern formation and it entails scrutinizing system behaviour in proximity to the initiation of pattern formation, where nonlinear influences are comparatively minor in relation to the broader dynamics. This method facilitates the development of simplified equations, typically achieved through perturbation techniques, which encapsulate the fundamental characteristics of pattern formation while disregarding higher-order nonlinearities. By concentrating on weakly nonlinear regimes, it is possible to glean insights into the mechanisms driving pattern formation without necessitating the resolution of the entire nonlinear system of equations, thereby rendering the analysis more manageable
\cite{BGMS2024, TLS2014, LBBGPS2017}.\\
The paper is organized as follows. In Section \ref{model} the problem formulation and the generalization of the model proposed in \cite{Du2017} into a PDE model are presented and the associated non-dimensional system is determined. The biologically meaningful equilibria are determined and discussed in Section \ref{equilibria}. In Sections \ref{preliminaries} and \ref{stability} the linear instability analysis of the coexistence equilibrium is performed, in order to determine conditions guaranteeing the stability of the endemic equilibrium in the absence of diffusion and diffusion-driven instability, namely Turing instability, in presence of self-diffusion and chemotaxis. In Section \ref{weakly} a weakly nonlinear analysis is performed to derive the Stuart-Landau equation for the amplitude. In Section \ref{numerical} the numerical investigations on the stability results are discussed and the governing equations are solved by means of a numerical scheme based on the finite difference method for both time and space variables. The paper ends with a concluding Section in which all the obtained results are summarised.

\section{Mathematical Model}\label{model}
Let us model tuberculosis in-host infection in a bounded domain $\Omega\subset\R^N \ (N=2,3)$. The model proposed by \cite{Du2017} is
\begin{equation}
    \begin{cases}\label{ODE-1}
    M'_u(\tilde{t})=s_M-\mu_MM_u(\tilde{t})-\beta M_u(\tilde{t})B(\tilde{t}), \\[2mm]
    M'_i(\tilde{t})=\beta M_u(\tilde{t})B(\tilde{t})-bM_i(\tilde{t})-\gamma M_i(\tilde{t})\dfrac{\frac{T(\tilde{t})}{M_i(\tilde{t})}}{\frac{T(\tilde{t})}{M_i(\tilde{t})}+c}, \\[2mm]
    B'(\tilde{t})=\delta B(\tilde{t})\left(1-\dfrac{B(\tilde{t})}{K}\right)+M_i(\tilde{t})\left(N_1b+N_2\gamma\dfrac{\frac{T(\tilde{t})}{M_i(\tilde{t})}}{\frac{T(\tilde{t})}{M_i(\tilde{t})}+c}\right)-M_u(\tilde{t})B(\tilde{t})(\eta+N_3\beta), \\[5mm]
    T'(\tilde{t})=s_T+\dfrac{c_MM_i(\tilde{t})T(\tilde{t})}{e_MT(\tilde{t})+1}+\dfrac{c_BB(\tilde{t})T(\tilde{t})}{e_BT(\tilde{t})+1}-\mu_TT(\tilde{t}).
    \end{cases}
   \end{equation}
where $M_u(\textbf{x},t)$ represents the uninfected macrophage population concentration, $M_i(\textbf{x},t)$ the infected macrophage population concentration, $B(\textbf{x},t)$ the bacteria population concentration and $T(\textbf{x},t)$ the CD4 T-cells population concentration. The uninfected macrophages are hired at constant rate $s_M$ and die naturally at constant rate $\mu_M$. Moreover, they become infected by bacteria at rate $\beta$. As regards the equation describing the evolution over the time of infected macrophages, they increase due to the number of new infections in the population $M_i$, and they are removed due either to the reproduction of Mtb inside of them, represented by the term $bM_i$, and to the killing at rate $\gamma$ by cell-mediated immunity, with saturating factor $c$. The bacterial population is divided into the extracellular and the intracellular bacteria. The former ones are killed by uninfected macrophages, while the latter ones die when a host macrophage is killed by apoptosis. The extracellular bacteria population grows logistically with constant rate $\delta$ and carrying capacity $K$. Due to macrophage bursting or death, intracellular bacteria become extracellular bacteria, and this is represented by the term $N_1bM_i$, with $N_1$ the average of bacteria released by a single host macrophage. Moreover, intracellular bacteria become extracellular bacteria also when a macrophage die by apoptosis, and this is denoted by the term $N_2\gamma M_i(T/M_i)/(T/M_i+c)$, where $N_2$ is the average of bacteria released after apoptosis, with $N_2<N_1$. The loss of extracellular bacteria is due either to the fact that they are killed by uninfected macrophages at rate $\eta$, and to the fact that by infecting the macrophages they become intracellular bacteria, and this is represented by the term $N_3\beta M_uB$, where $N_3$ is the average of bacteria carried by a single infected macrophage. The CD4 T-cells are recruited at constant rate $s_T$ and die at constant rate $\mu_T$. The infected macrophages stimulate the activation of the CD4 T-cells at rate $c_M$ with saturating factor $e_M$. Furthermore, CD4 T-cells can proliferate due to the presence of the bacteria at rate $c_B$ with saturating factor $e_B$.
   
The basic reproductive number of an infection, denoted by $\mathcal{\tilde{R}}_0$, is the expected number of new cases of infection directly generated by one infectious individual introduced into a population of all susceptible individuals. For the model of \cite{Du2017} $\mathcal{\tilde{R}}_0$ is given by

\begin{equation}
    {\mathcal{\tilde{R}}_0}=\dfrac{\delta}{(\eta+N_3\beta)M_u^0}+\dfrac{N_1 b+N_2\gamma}{b+\gamma}\dfrac{\beta}{\eta+N_3\beta},
\end{equation}  
with $M_u^0=\frac{s_M}{\mu_M}$ being the steady state for uninfected macrophages in a totally susceptible cell population. \\ \\
The cell-mediated immunity is driven by T-lymphocytes, which, at the site of infection, are mostly cytotoxic (CTLs) and CD4. As described in \cite{Du2017,WK2001,MK2004,GammackEtAl2005,SudEtAl2006}, the action of CD4 T-cells can maximize the killing efficacy of the macrophages and the proliferation of CTLs, therefore the action of CTLs can be considered proportional to the CD4 T-cell effector function (mathematically represented by $T/M_i$). Hence, the asymptotic behaviour of the immune response is described either as $T/M_i\ll c$ or $T/M_i\gg c$ and the infected macrophage killing term $\gamma M_i(t)\frac{(T(t)/M_i(t))}{(T(t)/M_i(t))+c}$ can be modified as $\gamma^* M_i$, where $\gamma^* \in \{0,\gamma\}$ (see \cite{Du2017}). This means that model \eqref{ODE-1} can be reduced to:
   \begin{equation}
    \begin{cases}\label{ODE-2}
    M'_u(\tilde{t})=s_M-\mu_MM_u(\tilde{t})-\beta M_u(\tilde{t})B(\tilde{t}), \\[2mm]
    M'_i(\tilde{t})=\beta M_u(\tilde{t})B(\tilde{t})-bM_i(\tilde{t})-\gamma^* M_i(\tilde{t}), \\[2mm]
    B'(\tilde{t})=\delta B(\tilde{t})\left(1-\dfrac{B(\tilde{t})}{K}\right)+\left(N_1b+N_2\gamma^*\right)M_i(\tilde{t})-(\eta+N_3\beta)M_u(\tilde{t})B(\tilde{t}).
    \end{cases}
   \end{equation}
   
    The basic reproductive number of the reduced model \eqref{ODE-2} is given by 
    \begin{equation}
    \label{Rasy}
    \tilde{\mathcal{R}}^{asy}_0=\dfrac{\delta}{(\eta+N_3\beta)M_u^0}+\dfrac{N_1 b+N_2\gamma^*}{b+\gamma^*}\dfrac{\beta}{\eta+N_3\beta},
\end{equation}
where $E_0=(M_u^0,0,0)$ is the disease-free equilibrium of \eqref{ODE-2}. \\
\noindent
Starting from system \eqref{ODE-2}, we allow a random motion (i.e. linear diffusion) and the chemotaxis (i.e. nonlinear diffusion) of the considered populations $M_u, M_i, B$, because for tuberculosis infection the production of chemokine, produced by extracellular bacteria and located in the infection site, represent an attractant for the uninfected macrophages. Hence the reaction-diffusion model that generalizes model \eqref{ODE-2} in $\Omega$, allowing self-diffusion and chemotaxis, is given by 
   \begin{equation}
    \begin{cases}\label{PDE}
    \dfrac{\partial M_u}{\partial \tilde{t}}=s_M-\mu_M M_u-\beta M_u B+ \tilde{\gamma_1} \Delta M_u - \tilde{\chi} \nabla \cdot (M_u \nabla M_i), \\[2mm]
    \dfrac{\partial M_i}{\partial \tilde{t}}=\beta M_u B -bM_i -\gamma^* M_i + \tilde{\gamma_2} \Delta M_i, \\[2mm]
    \dfrac{\partial B}{\partial \tilde{t}}=\delta B \left(1-\dfrac{B}{K}\right)+\left(N_1b+N_2\gamma^*\right)M_i -(\eta+N_3\beta)M_u B + \tilde{\gamma_3} \Delta B,
    \end{cases}
   \end{equation}
   under the following no-flux boundary conditions  
   \begin{equation}\label{BC}
       \nabla M_u \cdot \textbf{n}  = \nabla M_i \cdot \textbf{n} = \nabla B \cdot \textbf{n}= 0, \quad \text{on} \ \partial \Omega\times \R^+,
   \end{equation}
$\partial\Omega$ and $\textbf{n}$ being the boundary of $\Omega$ and the outward unit normal of $\partial\Omega$, respectively. Let us underline that $\tilde{\chi} \nabla \cdot (M_u \nabla M_i)$ is the chemotactic term and the chemokine activity is modelled by the model parameters. In particular, we assume $\tilde{\chi}$ to be positive since it models the chemoattractant activity of the chemokine. \\ Moreover, to system \eqref{PDE} we append the following smooth positive initial conditions
\begin{equation}
\label{IC}
    M_u(\tilde{\textbf{x}},0)=M_{u,0}(\tilde{\textbf{x}}), \ M_i(\tilde{\textbf{x}},0)=M_{i,0}(\tilde{\textbf{x}}), \ B(\tilde{\textbf{x}},0)=B_0(\tilde{\textbf{x}}) \quad \tilde{\textbf{x}}\in\Omega. 
\end{equation}
The model parameters are $\tilde{\Lambda}=\{s_M, \mu_M,b,\beta,\eta,\delta,K,\gamma^*,N_1,N_2,N_3,\tilde{\gamma_1},\tilde{\gamma_2},\tilde{\gamma_3},\tilde{\chi} \}\in\R_+^{15}$ and their biological meanings and values are reported in Table \ref{T1}. 
   \begin{table}
    \centering
     \resizebox{1\textwidth}{!}{
\begin{tabular}{|c|l|l|c|c|}
\hline
\textbf{Name} & \textbf{Meaning} & \textbf{Unit} & \textbf{Mean value (Range)} & \textbf{References} \\
\hline
$s_M$ & Recruitment rate of $M_u$ & 1/ml day & $5000 \ ([3300,7000])$ & \cite{WK2001,MK2004,GammackEtAl2005} \\
\hline
$s_T$ & Recruitment rate of $T$ & 1/ml day & $6.6 \ ([0.33,33])$ & \cite{WK2001,MK2004,GammackEtAl2005}  \\
\hline
$\mu_M$ & Loss rate of $M_u$ & 1/day & $0.01 \ ([0.01,0.011])$ & \cite{WK2001,MK2004,GammackEtAl2005}  \\
\hline
$b$ & Loss rate of $M_i$ & 1/day & $0.11 \ ([0.05,0.5])$ & \cite{WK2001,MK2004,GammackEtAl2005} \\
\hline
$\mu_T$ & Death rate of $T$ & 1/day & $0.33 \ ([0.05,0.33])$ & \cite{WK2001,MK2004,GammackEtAl2005} \\
\hline
$\beta$ & Infection rate by $B$ & ml/day & $5\times 10^{-6} \ ([10^{-8},10^{-5}])$ & \cite{WK2001,MK2004,GammackEtAl2005}  \\
\hline
$\gamma$ & Cell-mediated immunity rate & 1/day & $1.5 \ ([0.1,2])$ & \cite{WK2001,MK2004,GammackEtAl2005}  \\
\hline
$\eta$ & Bacteria killing rate by $M_u$ & ml/day & $1.25\times 10^{-9} \ ([1.25\times 10^{-9},1.25\times 10^{-7}])$ & \cite{WK2001,MK2004,GammackEtAl2005} \\
\hline
$\delta$ & Proliferation rate of $B$ & 1/day & $5\times 10^{-4} \ ([0,0.26])$ & \cite{WK2001,MK2004,GammackEtAl2005} \\
\hline
$c_M$ & Expansion rate of $T$ induced by $M_i$ & 1/day & $10^{-3} \ ([10^{-8},1])$ & \cite{WK2001,MK2004,GammackEtAl2005} \\
\hline
$c_B$ & Expansion rate of $T$ induced by $B$ & 1/day & $5\times 10^{-3} \ ([10^{-8},1])$ & \cite{WK2001,MK2004,GammackEtAl2005}  \\
\hline
$e_M$ & Saturating factor of $T$ expansion &  & $10^{-4}\ ([10^{-6},10^{-2}])$ & \cite{Du2017,ZFH2021}  \\
& related to $M_i$ & & &  \\
\hline
$e_B$ & Saturating factor of $T$ expansion &  & $10^{-4} \ $([$10^{-6}$,$10^{-2}$]) & \cite{Du2017,ZFH2021} \\
& related to $B$ & & & \\
\hline
$c$ & Half saturation ratio for $M_u$ lysis & $T/M_i$ & $3 \ ([0.3,30])$ & \cite{Du2017,ZFH2021}  \\
\hline
$K$ & Carrying capacity of $B$ & 1/ml & $10^8 \ $([$10^6$,$10^{10}$]) & \cite{Du2017,ZFH2021} \\
\hline
$N_1$ & Max MOI of $M_i$ & $B/M_i$ & $50 \ ([50,100])$ & \cite{WK2001,MK2004,GammackEtAl2005} \\
\hline
$N_2$ & Max number of $B$ released by & $T/M_i$ & $20 \ ([20,30])$ & \cite{WK2001,MK2004,GammackEtAl2005} \\
& apoptosis & & &  \\
\hline
$N_3$ & $N_3=N_1/2$ & $B/M_i$ & $25 \ ([25,50])$ & \cite{WK2001,MK2004,GammackEtAl2005}  \\
\hline
$\tilde{\gamma_1}$ & $M_u$ diffusion coefficient & m$^2$/day & $10^{-15} \ $([$10^{-15},10^{-9}$]) & \cite{GammackEtAl2004,natalini2010} \\
\hline
$\tilde{\gamma_2}$ & $M_i$ diffusion coefficient & m$^2$/day & $10^{-15} \ $([$10^{-15},10^{-9}$]) & \cite{GammackEtAl2004,natalini2010} \\
\hline
$\tilde{\gamma_3}$ & $B$ diffusion coefficient & m$^2$/day & $10^{-13} \ $([$10^{-14},10^{-13}$]) & \cite{GammackEtAl2004,natalini2010} \\
\hline
$\tilde{\chi}$ & Chemotaxis coefficient & m$^2$ ml/day & $10^{-3}$ \ $([10^{-3},10^3])$ & \cite{GammackEtAl2004,natalini2010} \\
\hline
\end{tabular} }
\caption{Biological meanings of the model parameters.}
\label{T1}
\end{table}
\\
In order to derive the dimensionless governing equations, let us introduce the following variables:
\begin{equation}\label{nodim}
\begin{aligned}
    &B=Kw, \ M_u=\frac{s_M}{\delta}u, \ M_i=\frac{\beta Ks_M}{\delta^2}i, \ \tilde{t}=\frac{1}{\delta}t, \ \tilde{\textbf{x}}=L\textbf{x}, \ \tilde{\gamma_i}=\delta L^2\gamma_i \quad (i=1,2,3), \\
    &\tilde{\chi}=\frac{\delta L^2}{K }\chi, \ \mu=\frac{\mu_M}{\delta}, \ n=\frac{\beta K}{\delta}, \ k=\frac{b}{\delta}, \ g=\frac{\gamma^*}{\delta}, \ h=\frac{\eta K}{\delta}, \ \alpha_1=\frac{\beta s_M}{\delta^2}, \ \alpha_2=\frac{s_M}{\delta K},
\end{aligned}
\end{equation}
where $L$ is the diameter of $\Omega$. Then, model \eqref{PDE}--\eqref{IC} becomes
   \begin{equation}
    \begin{cases}\label{PDE-2}
    \dfrac{\partial u}{\partial t}=1-\mu u-nuw+ \gamma_1 \Delta u - \chi \alpha_1 \nabla \cdot (u \nabla i), \\[2mm]
    \dfrac{\partial i}{\partial t}=uw -(k+g)i + \gamma_2 \Delta i, \\[2mm]
    \dfrac{\partial w}{\partial t}=w(1-w)+\alpha_1(N_1k+N_2g)i -\alpha_2(h+N_3n)uw + \gamma_3 \Delta w.
    \end{cases}
   \end{equation}
      \begin{equation}\label{BC-2}
       \nabla u \cdot \textbf{n}= \nabla i \cdot \textbf{n} = \nabla w \cdot \textbf{n} = 0, \quad \text{on} \ \partial \Omega\times \R^+,
   \end{equation}
   \begin{equation}
\label{IC-2}
    u(\textbf{x},0)=u_0(\textbf{x}), \ i(\textbf{x},0)=i_0(\textbf{x}), \ w(\textbf{x},0)=w_0(\textbf{x}) \quad \textbf{x}\in\Omega, 
\end{equation}
with $u_0(\textbf{x})>0, i_0(\textbf{x})>0, w_0(\textbf{x})>0$. According to \eqref{nodim}, the basic reproductive number \eqref{Rasy} becomes
\begin{equation}
    \mathcal{R}_0=\dfrac{\mu}{\alpha_2(h+N_3n)}+\dfrac{N_1k+N_2g}{k+g}\dfrac{n}{h+N_3n}.
\end{equation}

Let us denote by $\Lambda=\{ \mu,n,h,k,g,\alpha_1,\alpha_2,N_1,N_2,N_3,\gamma_1,\gamma_2,\gamma_3,\chi\}\in\R_+^{14}$ the set of dimensionless biological parameters.  

\section{Biologically equilibria: existence and a priori estimates}\label{equilibria}

The biologically meaningful equilibria are the steady states of \eqref{PDE-2}--\eqref{IC-2}, that are the non-negative solutions of the system
\begin{equation}
    \begin{cases}\label{EQ}
    1-\mu u-nuw+ \gamma_1 \Delta u - \chi \alpha_1 \nabla \cdot (u \nabla i)=0, \\[2mm]
    uw -(k+g)i + \gamma_2 \Delta i=0, \\[2mm]
    w(1-w)+\alpha_1(N_1k+N_2g)i -\alpha_2(h+N_3n)uw + \gamma_3 \Delta w=0.
    \end{cases}
   \end{equation}

For the sake of completeness, let us recall the results of \cite{Du2017} regarding the existence of meaningful constant equilibria. \\
The boundary equilibrium is the \emph{disease-free} equilibrium $E_0=\left(\frac{1}{\mu},0,0\right)$. The interior equilibrium $\bar{E}=(\bar{u},\bar{i},\bar{w})$ with non-null components represents the endemic equilibrium and is given by
\begin{equation}
    \bar{u}=\dfrac{1}{\mu+n\bar{w}}, \quad \bar{i}=\dfrac{\bar{w}}{(k+g)(\mu+n\bar{w})},
\end{equation}
where $\bar{w}$ is a positive solution of the equation
\begin{equation}
    a_2w^2+a_1w+a_0=0,
\end{equation}
with
\begin{equation}
\begin{aligned}
    &a_0=\alpha_2(h+N_3n)(1-\mathcal{R}_0), \\
    &a_1=\mu-n, \\
    &a_2=n.
\end{aligned} 
\end{equation}
According to \cite{Du2017}, the following theorem holds.  
\begin{thm}\label{esist}
     System \eqref{PDE-2}-\eqref{IC-2} has:
    \begin{enumerate}
        \item a unique endemic equilibrium if $\mathcal{R}_0>1$;
        \item a unique endemic equilibrium if $\mu < n$, and $\mathcal{R}_0=1$ or $\mathcal{R}_0=\mathcal{R}_c$, with $\mathcal{R}_c=1-\displaystyle\frac{(\mu-n)^2}{4 n \alpha_2 (h + N_3 n)}$;
        \item two endemic equilibria if $\mu<n$,  and $\mathcal{R}_c<\mathcal{R}_0<1$;
        \item no endemic equilibrium otherwise.
    \end{enumerate}
\end{thm}

\section{Preliminaries to stability}\label{preliminaries}
Let us consider a generic constant equilibrium $E^*=(u^*,i^*,w^*)$, and let us introduce the perturbation fields $\{U=u-u^*,I=i-i^*,W=w-w^*\}$. The perturbed system is given by
\begin{equation}
\label{PE}
    \begin{cases}
        \dfrac{\partial U}{\partial t}=a_{11}U+a_{12}I+a_{13}W+\gamma_1\Delta U-u^*\chi \alpha_1 \Delta I+F_1, \\[2mm]
        \dfrac{\partial I}{\partial t}=a_{21}U+a_{22}I+a_{23}W+\gamma_2\Delta I+F_2, \\[2mm]
        \dfrac{\partial W}{\partial t}=a_{31}U+a_{32}I+a_{33}W+\gamma_3\Delta W+F_3, 
    \end{cases}
\end{equation}
where
\begin{equation}\label{aij}
\begin{aligned}
    &a_{11}=-\mu-nw^*(<0), \quad a_{12}=0, \quad a_{13}=-nu^*(<0), \\
    &a_{21}=w^*(\geq 0), \quad a_{22}=-k-g(<0), \quad a_{23}=u^*(>0), \\
    &a_{31}=-\alpha_2(h+N_3n)w^*(\leq 0), \quad a_{32}=\alpha_1(N_1k+N_2g)(>0), \quad a_{33}=1-2w^*-\alpha_2(h+N_3n)u^*,
\end{aligned}    
\end{equation}
and the nonlinear terms $F_i(U,I,W)$, $i=1,2,3$ are
\begin{equation}
    F_1=-nUW-\chi \alpha_1 \nabla U\cdot \nabla I - \chi \alpha_1 U\Delta I, \quad F_2=UW, \quad F_3=-2W^2-\alpha_2(h+N_3n)UW.
\end{equation}
To system \eqref{PE} are associated the following no-flux boundary conditions
\begin{equation}
    \textbf{n}\cdot\nabla U=0, \quad \textbf{n}\cdot\nabla I=0, \quad \textbf{n}\cdot\nabla W=0, \quad \text{on} \ \partial \Omega\times \R^+. 
\end{equation}
\begin{rem}
  Let us remark that, in view of (\ref{EQ})$_3$, it turns out that
\begin{equation}
-w^*-\alpha_2(h+N_3 n)u^*= -1-\alpha_1 (N_1 k+ N_2 g)\dfrac{i^*}{w^*}.
\end{equation}
Then, from (\ref{aij})$_9$
\begin{equation}
 a_{33}=1-w^*-1-\alpha_1 (N_1 k+ N_2 g)\dfrac{i^*}{w^*}<0.
\end{equation}  
\end{rem}

\section{Linear instability analysis}\label{stability}
In this Section, we investigate the linear instability of the endemic equilibrium $\bar{E}$. We study the stability of the endemic equilibrium in the absence of diffusion and we look for conditions guaranteeing the occurrence of diffusion-driven instability, also called Turing instability. 
Moreover, for the sake of completeness, we recall the linear stability conditions determined in \cite{Du2017} both for boundary and interior equilibria ($E_0$ and $\bar{E}$) in the absence of diffusion.

\subsection{Linear instability in absence of diffusion}
In absence of diffusion, i.e. for $\gamma_i=0$ $(i=1,2,3)$ and $\chi=0$, the linear system associated to system \eqref{PE} evaluated around $\bar{E}$ is 
\begin{equation}\label{sist-matrND1}
    \dfrac{\partial \textbf{X}}{\partial t} = \mathcal{L}^0 \textbf{X}, 
\end{equation}
where
\begin{equation}\label{sist-matrND2}
    \begin{aligned}
        & \textbf{X}=(U,V,W), \qquad \mathcal{L}^0=\begin{pmatrix}
            a_{11} & 0 & a_{13} \\
            a_{21} & a_{22} & a_{23} \\
            a_{31} & a_{32} & a_{33} \\
        \end{pmatrix},
    \end{aligned}
\end{equation}
are the state vector and the Jacobian matrix evaluated in the coexistence equilibrium $\bar{E}$, according to \eqref{aij}.
The characteristic equation associated to $\mathcal{L}^0$ is
\begin{equation}\label{eq-carND}
    \lambda^3-I_1^0\lambda^2+I_2^0 \lambda-I_3^0=0,
\end{equation}
where $I_i^0$ ($i=1,2,3$) are the principal invariants in absence of diffusion, namely:
\begin{equation}
    \begin{array}{l}
        I_1^0=\text{tr}(\mathcal{L}^0)=a_{11}+a_{22}+a_{33}, \\
     I_2^0=\left|\begin{array}{cc}
     a_{11}&0\\
     a_{21}&a_{22}
     \end{array}\right| \!+\! \left|\begin{array}{cc}
     a_{11}&a_{13}\\
     a_{31}&a_{33}
     \end{array}\right| \!+\! \left|\begin{array}{cc}
     a_{22}&a_{23}\\
     a_{32}&a_{33}
     \end{array}\right|
=         a_{11}a_{22} \!+\!  a_{11}a_{33} \!-\! a_{13}a_{31} \!+\! a_{22}a_{33} \!-\! a_{23}a_{32}, \\[1.5mm]
I_3^0=\text{det}(\mathcal{L}^0)=a_{11}a_{22}a_{33}+a_{21}a_{32}a_{13}-a_{13}a_{22}a_{31}-a_{11}a_{23}a_{32}.
    \end{array}
\end{equation}
 The interior equilibrium $\Bar{E}$ is linearly stable if and only if the Routh-Hurwitz conditions are verified \cite{Murray2_2002,Murray1_2002,Merkin1997}:
    \begin{equation}\label{RH}
        I_1^0<0, \qquad I_3^0<0, \qquad I_1^0I_2^0-I_3^0<0. 
    \end{equation}
    Conditions \eqref{RH} obviously imply $I_2^0>0$. If at least one of \eqref{RH} is reversed, there exists at least one eigenvalue of $\mathcal{L}^0$ with positive real parts, and hence, in that case, the linear instability of $\Bar{E}$ is guaranteed.
\\
\noindent
According to \cite{Du2017}, in absence of diffusion the following theorems hold.
\begin{thm}
   The disease-free equilibrium $E_0$ of the system \eqref{PDE-2}-\eqref{IC-2} is locally asymptotically stable if and only if $\mathcal{R}_0<1$. Furthermore, $E_0$ is unstable if $\mathcal{R}_0>1$.
\end{thm}
\begin{thm}
    System \eqref{PDE-2}-\eqref{IC-2} have:
    \begin{enumerate}
        \item two endemic equilibria, $E_+$ and $E_-$, if $\mathcal{R}_c<\mathcal{R}_0<1$. $E_+$ is locally asymptotically stable and $E_-$ is unstable;
        \item a unique endemic equilibrium $\hat{E}$ if $\mathcal{R}_0>1$. $\hat{E}$ is locally asymptotically stable.
    \end{enumerate}
\end{thm}

\subsection{Chemotaxis driven linear instability}
In order to find conditions guaranteeing the occurrence of Turing instability, i.e. conditions guaranteeing that, $E_c$ stable in the homogeneous case, loses its stability in the heterogeneous case,  in this Section we consider simultaneous linear self-diffusion and chemotaxis. The linear system associated to the system \eqref{PE}, is given by

\begin{equation}\label{sist-matrC1}
    \dfrac{\partial \textbf{X}}{\partial t} = \mathcal{L}^0 \textbf{X} + \mathcal{D} \Delta \textbf{X}, 
\end{equation}
where $\textbf{X}$ and $\mathcal{L}^0$ are defined in \eqref{sist-matrND2}, while $\mathcal{D}$ is 
\begin{equation}
    \mathcal{D}=\begin{pmatrix}
            \gamma_{1} & -\bar u\chi \alpha_1 & 0 \\
            0 & \gamma_{2} & 0 \\
            0 & 0 & \gamma_{3} \\
        \end{pmatrix}
\end{equation}
with $\det \mathcal{D}=\gamma_1\gamma_2\gamma_3>0$.
Setting
\begin{equation}\label{set}
         \Gamma=\gamma_1\gamma_2+\gamma_1\gamma_3+\gamma_2\gamma_3(>0);\quad 
         D_\gamma=-\left[
         \gamma_1(a_{22}+a_{33})+\gamma_2(a_{11}+a_{33})+\gamma_3(a_{11}+a_{22})\right](>0),
   \end{equation}
the dispersion relation associated to \eqref{sist-matrC1} governing the eigenvalues $\lambda$ in terms of the wavenumber $k$, is
\begin{equation}\label{disp-relC1}
    \lambda^3-T_k(k^2)\lambda^2+I_2(k^2)\lambda-h(k^2)=0,
\end{equation}
where
\begin{equation}\label{disp-rel}
\mkern-18mu \mkern-9mu
    \begin{aligned}
       T_k(k^2)& \!=\! \text{tr}(\mathcal{L}^0)-k^2\text{tr}(\mathcal{D})=I^0_1-k^2(\gamma_1+\gamma_2+\gamma_3), \\
       I_2(k^2)& \!=\!\Gamma k^4 \!+\! k^2[D_{\gamma}- a_{21}\chi \alpha_1 \bar u] + I^0_2,\\[1.5mm]
        h(k^2)& \!=\! -k^6\text{det}(\mathcal{D})+k^4[a_{11}\gamma_2\gamma_3+a_{22}\gamma_1\gamma_3+a_{33}\gamma_1\gamma_2+a_{21}\gamma_3\alpha_1\chi \bar u]\\ &
\!-\! k^2[\gamma_1(a_{22}a_{33} \!-\! a_{23}a_{32}) \!+\! \gamma_2(a_{11}a_{33} \!-\! a_{13}a_{31}) \!+\! a_{11}a_{22}\gamma_3 \!+\! \alpha_1\chi(a_{21}a_{33} \!-\! a_{31}a_{23})\bar u]+I_3^0.
    \end{aligned}
\end{equation}
The following theorem holds.\begin{thm}
When $\bar{E}$ exists, if \eqref{RH} holds together with 
   \begin{equation}\label{TUR}
     \chi>\chi_c:=\displaystyle\frac{D_\gamma+2\sqrt{\Gamma I^0_2}}{\alpha_1 a_{21} \bar u}       
   \end{equation}
   then Turing instability occurs.
\end{thm}
\begin{proof}
Conditions \eqref{RH} imply that $\bar E$ is linearly homogeneously stable. 
Since $I^0_1<0$, $I^0_2>0$ one has $T_k(k^2)<0, \forall k^2$. In order to find conditions guaranteeing the occurrence of diffusion-driven instability, one of the following conditions has to be reversed for at least a $\bar k\in\mathbb R^+$:
    \begin{equation}\label{RH2}
I_2(k^2)>0,\quad    h(k^2)<0, \quad T_k(k^2) I_2(k^2) - h(k^2)<0.
    \end{equation}
From \eqref{disp-rel}$_2$ one has that, a \emph{necessary condition} guaranteeing that (\ref{RH2})$_1$ can be reversed, for at least some $k^2\in\mathbb R^+$, is that 
\begin{equation}\label{cnec}
\chi>\bar\chi:=\displaystyle\frac{D_{\gamma}}{\alpha_1 a_{21}\bar u}.
\end{equation}
In view of 
$ \displaystyle\frac{\partial I_2}{\partial k^2}=2\Gamma k^2+(D_{\gamma}-\alpha_1 a_{21} \bar u \chi)$, if (\ref{cnec}) is satisfied, $I_2(k^2)$ attains a minimum in $k^2_c:=\displaystyle\frac{\alpha_1 a_{21} \bar u\chi-D_{\gamma}}{2\Gamma}$, given by
\begin{equation}\label{minI2}
    I_2(k^2_c)=-\displaystyle\frac{(\alpha_1 a_{21}\bar u\chi-D_{\gamma})^2}{4\Gamma}+I^0_2.
\end{equation}
Simple calculation shows that (\ref{TUR}) guarantees that $I_2(k^2_c)<0$ and hence Turing instability occurs.
\end{proof}
\begin{figure}[ht!]
    \centering
    \includegraphics[scale=0.6]{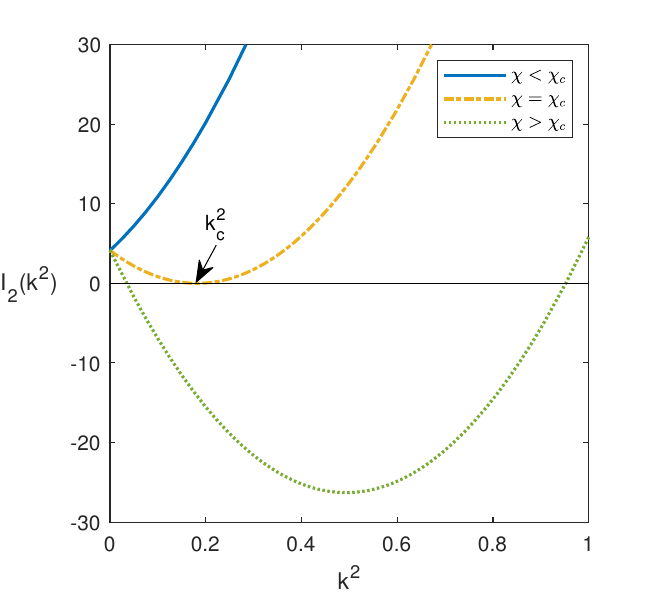}
    \caption{Function $I_2(k^2)$ plotted with respect to $k^2$ for quoted values of the bifurcation parameter $\chi$.}
    \label{I2k}
\end{figure}

According to the previous Theorem, in Figure \ref{I2k} $I_2(k^2)$ is plotted as function of $k^2$ for quoted values of $\chi$, which is our bifurcation parameter, in order to analyse how the chemotaxis parameter $\chi$ affects the function $I_2(k^2)$ and to show the range of unstable wavenumbers.

\section{Weakly nonlinear analysis}\label{weakly}
In this Section we shall derive the Stuart-Landau amplitude equation for the spatially periodic solutions to the system \eqref{PDE-2}. By using a multiple scale method we perform a weakly nonlinear analysis close to the uniform steady state $\bar E$. We set a small control parameter $\varepsilon^2=(\chi-\chi_c)/\chi_c$, describing the dimensionless distance between $\chi$ and $\chi_c$. Let us consider the perturbed system \eqref{PE}, that can be written as
\begin{equation} \label{wnl}
    \frac{\partial\textbf{v}}{\partial t}=\mathcal{L}^{\chi}\textbf{v}+\mathcal{N}\textbf{v},
\end{equation}
with $\textbf{v}=(U,I,W)^T$,
$\mathcal{L}^{\chi}=\mathcal{L}^0+\mathcal{D}\Delta$ is the linear operator and $\mathcal{N}$ is the operator which contains the contribution of the nonlinear terms given by
\begin{equation}
    \mathcal{N}\textbf{v}:=\left(\begin{array}{c}
         -\chi \alpha_1 \nabla U\cdot\nabla I-\chi\alpha_1 U\Delta I  \\
         0\\
         0
    \end{array}\right)+\left(\begin{array}{c}
         -nUW  \\
         UW\\
         -2W^2-\tilde\alpha_2 UW
    \end{array}\right)
\end{equation}
being $\tilde \alpha_2=\alpha_2(h+N_3n)$.
 We expand the bifurcation parameter $\chi$ and $\textbf{v}$ as 
\begin{equation} \label{exp}
\begin{aligned}
    &\textbf{v}=\varepsilon \textbf{v}_1+\varepsilon^2 \textbf{v}_2+\varepsilon^3 \textbf{v}_3+O(\varepsilon^4), \\
    & \chi=\chi_c+\varepsilon^2\chi_2+O(\varepsilon^4),
\end{aligned}
\end{equation}
and by assuming a multiple scale dependence $\mathbf{v}=\mathbf{v}(T_2,T_4,...)$, with $T_{2i}=\varepsilon^{2i}t$ for $i=1,2$, the time derivative can be expanded as
\begin{equation}
    \frac{\partial}{\partial t}=\frac{\partial}{\partial T_0}+\varepsilon^2\frac{\partial}{\partial T_2}+\varepsilon^4\frac{\partial}{\partial T_4}+O(\varepsilon^5). 
\end{equation}
The operator $\mathcal{L}^{\chi}$ can be expanded as follows
\begin{equation}
    \mathcal{L}^{\chi}=\mathcal{L}^{\chi_c}+\varepsilon^2\begin{pmatrix}
        0 & -\alpha_1\chi_2\bar u\Delta & 0 \\
        0 & 0 & 0 \\
        0 & 0 & 0
    \end{pmatrix}+O(\varepsilon^4),
\end{equation}
with 
\begin{equation}
    \mathcal{L}^{\chi_c}=\begin{pmatrix}
        a_{11}+\gamma_1\Delta & -\alpha_1\chi_c\bar u\Delta & a_{13} \\
        a_{21} & a_{22}+\gamma_2\Delta & a_{23} \\
        a_{31} & a_{32} & a_{33}+\gamma_3\Delta
    \end{pmatrix}.
\end{equation}
Now we substitute \eqref{exp} in \eqref{wnl} and we collect the terms at each order of $\varepsilon$. \\
At the first order $O(\varepsilon)$ we obtain
\begin{equation}\label{oeps}
    \mathcal{L}^{\chi_c}\textbf{v}_1=\mathbf{0}.
\end{equation}
By imposing the Neumann boundary conditions, we look for solutions as
\begin{equation}\label{eq42}
    \textbf{v}_1=\bm{\rho}A(T_2,T_4,...)\cos(k_cx), \quad \bm{\rho}=(\rho_u, \rho_i, \rho_w).   
\end{equation}
Substituting \eqref{eq42} in \eqref{oeps}, one obtains that system \eqref{oeps} is verified if 
\begin{equation}
    \bm{\rho}=\left(
\dfrac{\bar u \chi_c\alpha_1 k^2_c a_{23}-a_{13}(a_{22}-\gamma_2 k^2_c)}{(a_{22}-\gamma_2k^2_c)(a_{11}-\gamma_1k^2_c)-a_{21}\bar u \chi_c \alpha_1 k^2_c},\dfrac{a_{21}a_{13}-a_{23}(a_{11}-\gamma_1 k^2_c)}{(a_{22}-\gamma_2k^2_c)(a_{11}-\gamma_1k^2_c)-a_{21}\bar u \chi_c \alpha_1 k^2_c},   
     1\right).
\end{equation}
At the second order $O(\varepsilon^2)$ we obtain
\begin{equation}\label{oeps2}
    \mathcal{L}^{\chi_c}\mathbf{v}_2=\mathbf{F}(\mathbf{v}_1),
\end{equation}
with
\begin{equation}
    \mathbf{F}(\mathbf{v}_1)=\begin{pmatrix}
        nU_1W_1 \\
        -U_1W_1 \\
       2 W_1^2+\tilde\alpha_2 U_1W_1
    \end{pmatrix} +\chi_c\alpha_1\begin{pmatrix}
        \nabla U_1 \cdot \nabla I_1+U_1\Delta I_1 \\
        0 \\
        0
    \end{pmatrix}.
\end{equation}
By the Fredholm alternative, \eqref{oeps2} admits solutions if and only id the solvability condition is verified, i.e. if and only if $\langle \mathbf{F}, \tilde{\bm{\psi}} \rangle=0$, where $\tilde{\bm{\psi}}\in\ker{(\mathcal{L}^{\chi_c})^*}$, with $(\mathcal{L}^{\chi_c})^*$ the adjoint operator of $\mathcal{L}^{\chi_c}$. Since $\tilde{\bm{\psi}}=\bm{\psi}\cos(k_cx)$, with 
\begin{equation}
    \bm{\psi}=(\psi_u,\psi_i,\psi_w)=\left(
   \dfrac{a_{21}(a_{33}-\gamma_3k^2_c)-a_{31}a_{23}}{a_{23}(a_{11}-\gamma_1k^2_c)-a_{13}a_{21}},
      \dfrac{a_{13}a_{31}-(a_{11}-\gamma_1k^2_c)(a_{33}-\gamma_3 k^2_c)}{a_{23}(a_{11}-\gamma_1k^2_c)-a_{13}a_{21}},
   1\right),
\end{equation}
Fredholm alternative is automatically satisfied. By imposing the no-flux boundary conditions, the solution can be written as
\begin{equation}\label{eq47}
    \mathbf{v}_2=A^2[\bm{\mu}+\bm{\theta}\cos(2k_cx)], \quad \bm{\mu}=(\mu_u,\mu_i,\mu_w), \quad \bm{\theta}=(\theta_u,\theta_i,\theta_w),
\end{equation}
with $\mathbf \mu$ and $\mathbf \theta$ verifying
\begin{equation}\label{sistemalin}
    \begin{cases}
    (a_{11}-\gamma_1k^2_c)(\mu_u+\theta_u)+\bar u \chi_c\alpha_1 k^2_c (\mu_i+\theta_i)+a_{13}(\mu_w+\theta_w)=n\rho_u\rho_w-\chi_c\alpha_1 k^2_c\rho_u\rho_i,\\
     (a_{11}-\gamma_1k^2_c)(\mu_u-\theta_u)+\bar u \chi_c\alpha_1 k^2_c (\mu_i-\theta_i)+a_{13}(\mu_w-\theta_w)=\chi_c\alpha_1 k^2_c\rho_u\rho_i,\\
     a_{21}(\mu_u+\theta_u)+(a_{22}-\gamma_2k^2_c)(\mu_i+\theta_i)+a_{23}(\mu_w+\theta_w)=-\rho_u\rho_w,\\
     a_{21}(\mu_u-\theta_u)+(a_{22}-\gamma_2k^2_c)(\mu_i-\theta_i)+a_{23}(\mu_w-\theta_w)=0,\\
     a_{31}(\mu_u+\theta_u)+a_{32}(\mu_i+\theta_i)+(a_{33}-\gamma_3k^2_c)(\mu_w+\theta_w)=2\rho^2_w+\tilde\alpha_2\rho_u\rho_w,\\
      a_{31}(\mu_u-\theta_u)+a_{32}(\mu_i-\theta_i)+(a_{33}-\gamma_3k^2_c)(\mu_w-\theta_w)=0.
      \end{cases}
\end{equation}
At the third order $O(\varepsilon^3)$ one gets
\begin{equation}\label{oeps3}
    \mathcal{L}^{\chi_c}\mathbf{v}_3=\mathbf{G}(\mathbf{v}_1,\mathbf{v}_2),
\end{equation}
with
\begin{equation}\label{G}
\begin{aligned}
    \mathbf{G}(\mathbf{v}_1,\mathbf{v}_2)=\frac{\partial}{\partial T_2}\begin{pmatrix}
        U_1 \\
        I_1 \\
        W_1
    \end{pmatrix}&-\begin{pmatrix}
        \chi_2\alpha_1 \bar u k_c^2I_1 \\
        0 \\
        0
    \end{pmatrix}+\begin{pmatrix}
        n(U_1W_2+U_2W_1) \\
        -(U_1W_2+U_2W_1) \\
        4W_1W_2+\tilde\alpha_2(U_1W_2+U_2W_1)
    \end{pmatrix} + \\ 
    &+ \chi_c\alpha_1\begin{pmatrix}
        \nabla U_1\cdot \nabla I_2+ \nabla U_2\cdot \nabla I_1+U_1\Delta I_2+U_2\Delta I_1 \\
        0 \\
        0
    \end{pmatrix}.
\end{aligned}
\end{equation}
The right-hand side of \eqref{oeps3} can be rewritten as follows
\begin{equation}
    \mathbf{G}(\mathbf{v}_1,\mathbf{v}_2)=\left[\mathbf{G}^{(0)}\frac{\partial A}{\partial T_2}+\mathbf{G}^{(1)}A+\mathbf{G}^{(2)}A^3\right]\cos(k_cx)+\mathbf{G}^*\cos(3k_c x),
\end{equation}
with
\begin{equation}
    \begin{array}{l}
        \mathbf{G}^{(0)}=\begin{pmatrix}
            \rho_u \\
            \rho_i \\
            \rho_w 
        \end{pmatrix}, \quad
        \mathbf{G}^{(1)}=\left(\begin{array}{c}
            -\chi_2 \alpha_1\bar u k_c^2\rho_i \\
            0 \\
            0 
        \end{array}\right), \\
        \mathbf{G}^{(2)}=\left(\begin{array}{c}
           n\left[\rho_u\mu_w+\mu_u\rho_w+(\rho_u\theta_w+\rho_w\theta_u)/2\right]\\
           -\left[\rho_u\mu_w+\mu_u\rho_w+(\rho_u\theta_w+\rho_w\theta_u)/2\right]\\
           2(2\rho_w\mu_w+\rho_w\theta_w)+\tilde\alpha_2\left[\rho_u\mu_w+\mu_u\rho_w+(\rho_u\theta_w+\rho_w\theta_u)/2\right]
        \end{array}\right)\\
       \qquad \qquad
       + \chi_c\alpha_1 k^2_c\left(\begin{array}{c}
        -\rho_u\theta_i+\rho_i\theta_u/2-\rho_i\mu_u\\
        0\\
        0
        \end{array}\right).
    \end{array}
\end{equation}
 The solvability condition at this order leads to the Stuart-Landau equation for the amplitude
\begin{equation}\label{Landau}
    \dfrac{\partial A}{\partial T_2} = \sigma A - L A^3, 
\end{equation}
where
\begin{equation}\label{sigmaL}
    \sigma=-\frac{\langle\mathbf{G}^{(1)},\tilde{\bm{\psi}}\rangle}{\langle\mathbf{G}^{(0)},\tilde{\bm{\psi}}\rangle}, \quad L=\frac{\langle\mathbf{G}^{(2)},\tilde{\bm{\psi}}\rangle}{\langle\mathbf{G}^{(0)},\tilde{\bm{\psi}}\rangle}.
\end{equation}
The solutions of the Stuart--Landau equation \eqref{Landau} are the steady solution $A=0$ and $A_\infty=\pm \sqrt{\dfrac{\sigma}{L}}$. The stationary solution $A=0$ is stable if $\sigma<0$ and unstable if $\sigma>0$. When $A=0$ loses its stability, two equilibrium points exist and are stable (supercritical bifurcation). Then the solution reverts to one of the functions
\begin{equation}\label{super}
\mathbf{v}=\varepsilon {\bm{\rho}} A_\infty \cos (k_c x)+\varepsilon^2 A^2_\infty (\bm{\mu}+\bm{\theta} \cos (2k_cx)) + O(\varepsilon^3).
\end{equation}
If $L<0$, the weakly nonlinear analysis has to be extended up to the fifth order in $\varepsilon$. Up to the third order, we obtain the same equations as in the previous case. Then, if the solvability condition $<\bm{G},\tilde{\bm{\Psi}}>=0$ is satisfied, equation (\ref{oeps3}) admits the solution
\begin{equation}\label{sub3}
{\bm{v}_3}=(A{\bm{v}}_{31}+A^3{\bm{v}}_{32})\cos (k_c x)+A^3{\bm{v}}_{33}\cos (3k_cx),
\end{equation}
where $\bm{v}_{3i}, (i=1,2,3)$ verify
$$\begin{array}{l}
\mathcal L^{\chi_c}\mathbf v_{31}=\sigma \bm{\rho}+\bm{G}^{(1)}\\
\mathcal{L}^{\chi_c}\bm{v}_{32}=\bm{G}^{(2)}-L\bm{\rho},\\
\mathcal{L}^{\chi_c}\bm{v}_{33}=\bm{G}^{(*)}.
\end{array}
$$
At $O(\varepsilon^4)$, the solvability condition is automatically verified and the solution can be written as
$$
\mathbf v_4=A^2\mathbf v_{40}+A^4\mathbf v_{41}+(A^2\mathbf v_{42}+A^4\mathbf v_{43})\cos (2k_cx)+A^4 \mathbf v_{44}\cos (4k_cx),
$$
where $\mathbf v_{4i}, (i=0,\cdots, 4)$ are obtained by solving five linear systems.
\\
\noindent
Finally, at the fifth order, the Fredholm solvability condition leads to the Stuart-Landau equation
\begin{equation}\label{SL2}
\dfrac{\partial A}{\partial T_4}=\tilde{\sigma} A- \tilde{L} A^3+\tilde{Q} A^5,\end{equation}
where $\tilde \sigma, \tilde L, \tilde Q$ depend on the system parameters. Adding (\ref{Landau}) and (\ref{SL2}), one obtains the quintic Stuart-Landau equation
\begin{equation}\label{SL3}
\dis\frac{\partial A}{\partial T}=\bar\sigma A-\bar L A^3+\bar Q A^5,\end{equation}
being $\bar \sigma=\sigma+\varepsilon^2\tilde\sigma$.
The bifurcation analysis of equation (\ref{SL3}) leads to the following results.
There exists a threshold $\chi_s<\chi_c$ such that
\begin{itemize}
    \item if $\chi<\chi_s$ the only steady state is the origin (stable);
    \item if $\chi_s<\chi<\chi_c$, there are 5 equilibria. The origin (stable); the solutions  $A_{\infty,\pm}=\pm\sqrt{\dis\frac{\bar L-\sqrt{\bar L^2-4\bar\sigma \bar Q}}{2\bar Q}}$  (stable) and the solutions $A_{\infty,1,2}=\pm\sqrt{\dis\frac{\bar L+\sqrt{\bar L^2-4\bar\sigma \bar Q}}{2\bar Q}}$ (unstable);
    \item if $\chi>\chi_c$ there are 3 equilibria, the origin (unstable) and  $A_{\infty,\pm}$ (stable).
\end{itemize}
Then,
 for $\bar\sigma>0,\bar L<0, \bar Q<0$,  at the onset of instability (third case), there exists two real stable equilibria $A_{\infty,\pm}$ representing the asymptotic value of the amplitude.

\section{Numerical results}\label{numerical}
In this Section, we numerically investigate our theoretical results for several scenarios, to evaluate the influence of the biological parameters on the system dynamic.  We implemented a numerical scheme based on the finite difference method for both time and space variables on the MatLab software \cite{Garvie2007}.

According to Table \ref{T1} let us employ the following set of biological parameters $S_1=\{\mu=0.04,n=0.8,k=0.44,g=8,h=0.5,\alpha_1=0.0106,\alpha_2=0,0132,N_1=50,N_2=20,N_3=25\}$ together with $\gamma_1=\gamma_2=10^{-10}, \gamma_3=10^{-8}$ and $\chi=10^{-1}$. With this set of parameters, case 3 of Theorem \ref{esist} is verified, so there are two endemic equilibria, i.e.
\begin{equation}
    E^*_1=[23.7337,0.0075019,0.0026678], \ E^*_2=[1.2533,0.14068,0.94733]. 
\end{equation}
In absence of diffusion, for $E^*_1$ one gets $I_1^0=-13.9098$, $I_2^0=0.5931$, $I_3^0=0.0170$ and $I_1^0I_2^0-I_3^0=-8.2671$, hence conditions \eqref{RH} are not satisfied and $E^*_1$ is unstable. \\ For $E^*_2$, one gets $I_1^0=-10.4717$, $I_2^0=15.4569$, $I_3^0=-6.0424$ and $I_1^0I_2^0-I_3^0=-155.8169$, hence $E^*_2$ is locally stable in absence of diffusion. In particular in this scenario $\chi_c=1.3252\times10^{-4}$, for the endemic equilibrium $E_2^*$ conditions \eqref{RH}-\eqref{TUR} are both satisfied and the space distribution of the populations $u,i,w$, choosing as initial condition a random perturbation of $E_2^*$, are plotted in Figures \ref{E3tur}(a)-(c), after $200$ days in the spatial domain $[0,1] \times [0,1]$.

The next simulation set we adopted is 
$S_2=\{\mu=0.5,n=0.02,k=1,g=1.5,h=2,\alpha_1=0.001,\alpha_2=0.03,N_1=70,N_2=25,N_3=35 \}$
together with $\gamma_1=\gamma_2=10^{-9}, \gamma_3=10^{-13}$ and $\chi=10^{-1}$. In this configuration, case 1 of Theorem \ref{esist} is verified and the only endemic equilibrium is
\begin{equation}
    E^*_3=[1.9285,0.7149,0.9267].
\end{equation}
In absence of diffusion, for $E^*_3$ one gets $I_1^0=-4.0282$, $I_2^0=4.1338$, $I_3^0=-1.1979$ and $I_1^0I_2^0-I_3^0=-15.4536$, hence $E^*_3$ is locally stable. When self-diffusion and chemotaxis are taken into account, for $E_3^*$ conditions \eqref{RH}-\eqref{TUR} are both satisfied ($\chi_c=5.0945\times10^{-6}$), hence Turing patterns are expected. In Figures \ref{E5tur} the $2D$ spatial distribution of the numerical solutions of system \eqref{PDE-2}-\eqref{IC-2} are plotted, choosing as initial condition a random perturbation of $E_3^*$. In particular the \textit{complementary} spatial distributions of $u$, $i$ and $w$ are plotted for different time intervals, respectively. As one can see, with this set of parameters, the pattern formation happens right away after 10 days and continues to evolve with respect to time. In this configuration, the spatio-temporal evolution of uninfected macrophages $u$, infected macrophages $i$ and bacteria $w$ are plotted in Figure \ref{E5_ST}. 

In Figures \ref{E5_3_tur} the spatial distributions of the numerical solutions of \eqref{PDE-2}-\eqref{IC-2} are plotted  for $S_2 \cup \{ \gamma_1=\gamma_2=10^{-11}, \gamma_3=10^{-14}, \chi = 10^{-1} \}$, choosing a random perturbation of the endemic equilibrium $E_3^*$. In this scenario, one gets $\chi_c=5.0981\times10^{-8}$. Let us underline that from Figures \ref{E5tur} and \ref{E5_3_tur} one can observe a granuloma-like pattern. The scenarios depicted in Figures \ref{E5tur} and \ref{E5_3_tur} are different: in the first scenario one can observe a higher proliferation of uninfected macrophages, while in the second scenario a more even distribution of bacteria is appreciated, this is motivated by the fact that the value chosen for the diffusion coefficient of the macrophages is higher in the first configuration.

\begin{figure}[ht!]
    \centering
    \subfigure[]{
    \includegraphics[scale=0.35]{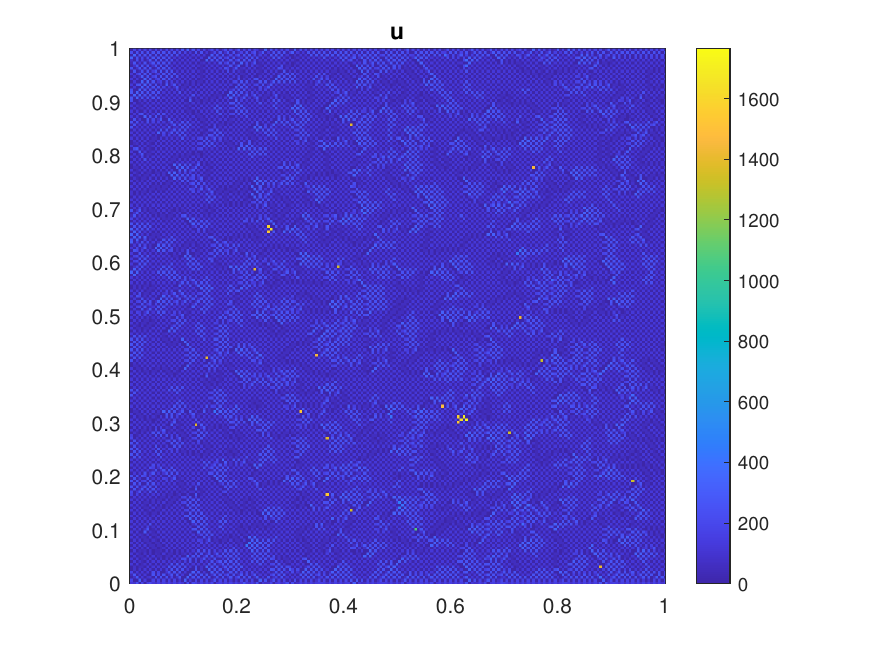}}
    \subfigure[]{
    \includegraphics[scale=0.35]{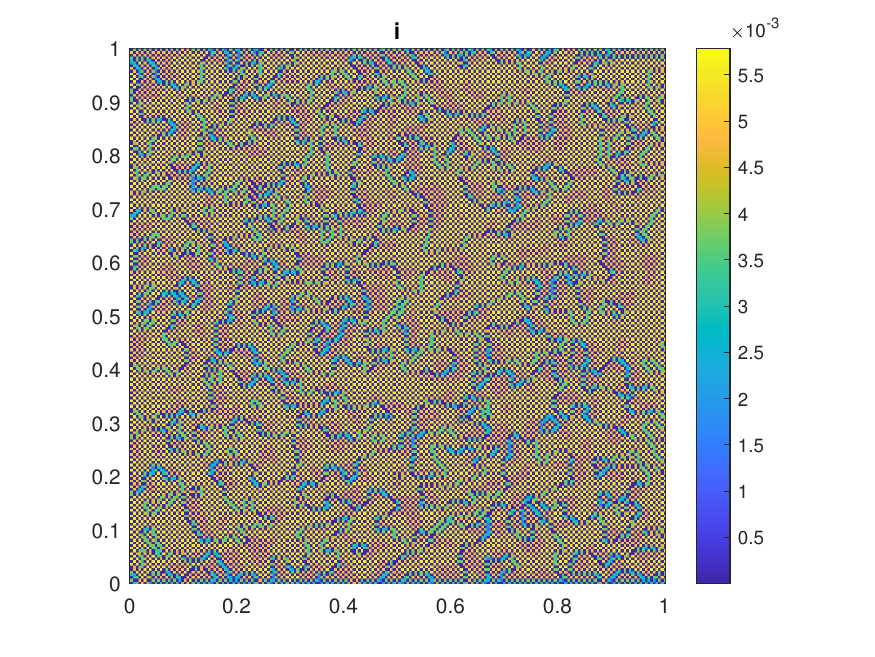}}
    \subfigure[]{
    \includegraphics[scale=0.37]{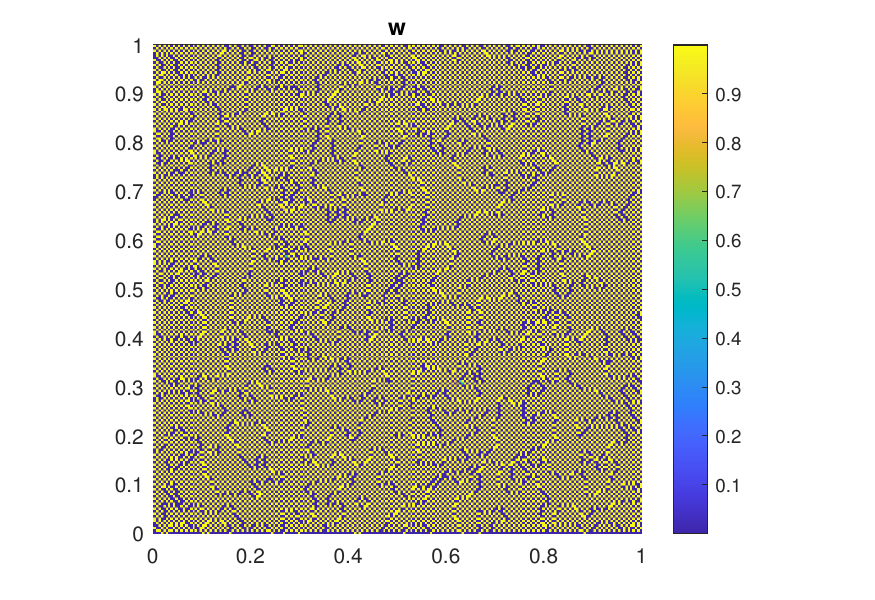}} 
    \caption{(a)-(c) Spatial distribution of the populations $u,i,w$ plotted choosing as initial conditions a random perturbation of $E^*_2$ and for parameters $S_1 \cup \{\gamma_1=\gamma_2=10^{-10}, \gamma_3=10^{-8},\chi=10^{-1} \}$, after $200$ days. In the spatial domain $[0,1] \times [0,1]$ with discretization time step $\Delta t=10^{-3}$, discretization space step $\Delta x=\Delta y=0.005$.}
    \label{E3tur}
\end{figure}

\begin{figure}[ht!]
    \centering
    \subfigure[10 days]{
    \includegraphics[scale=0.35]{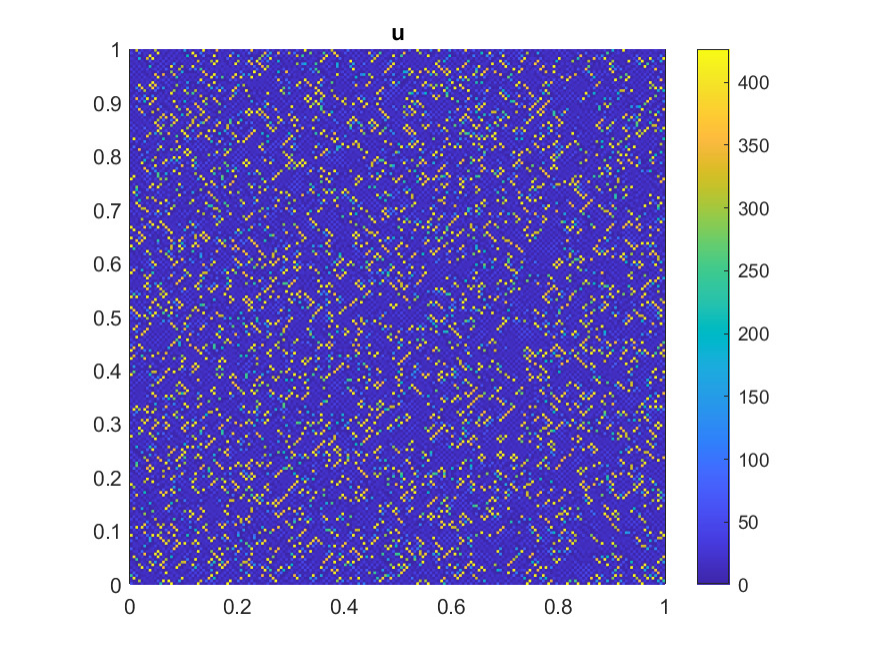}}
    \subfigure[10 days]{
    \includegraphics[scale=0.35]{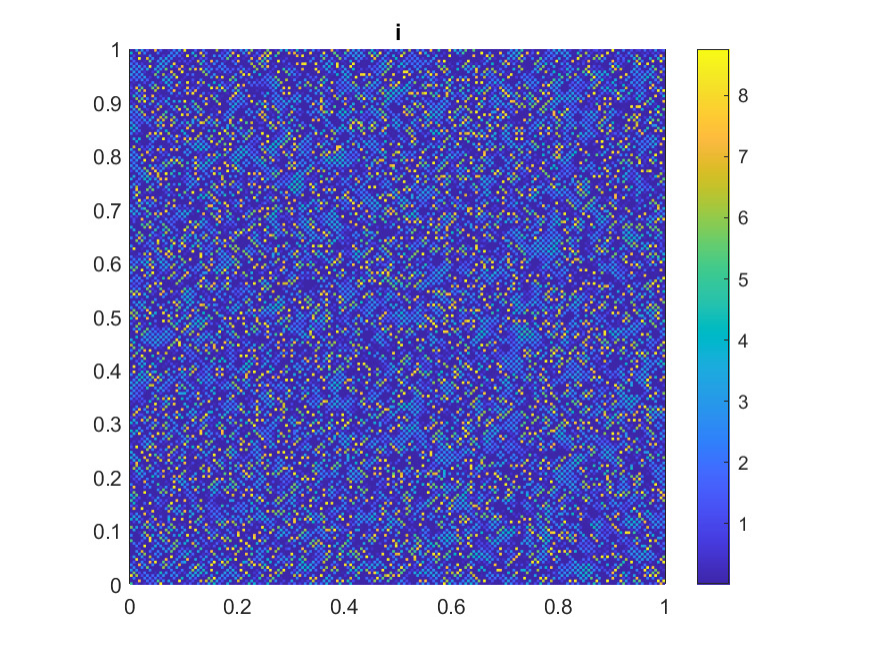}}
    \subfigure[10 days]{
    \includegraphics[scale=0.35]{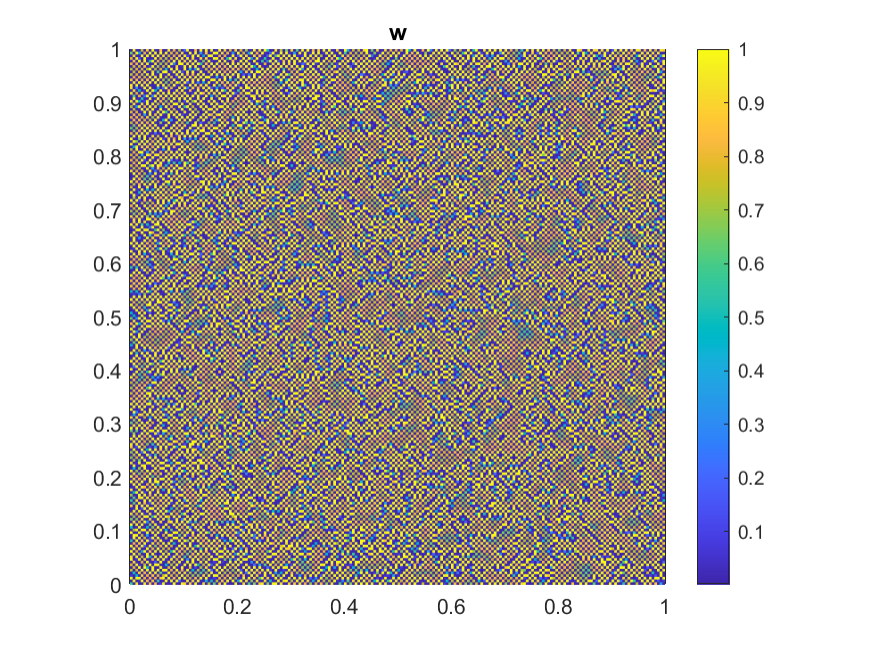}} \\
    \subfigure[20 days]{
    \includegraphics[scale=0.35]{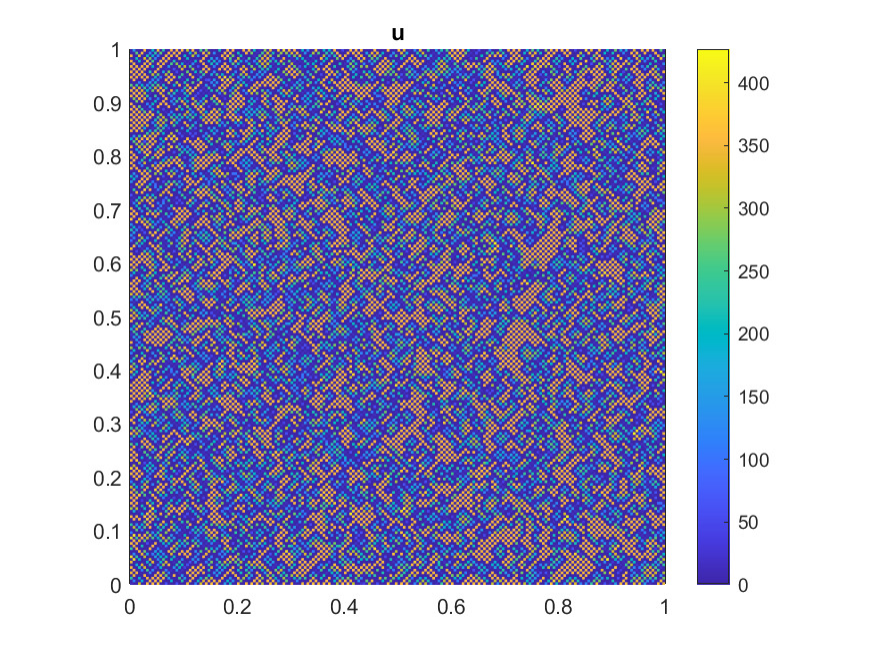}}
    \subfigure[20 days]{
    \includegraphics[scale=0.35]{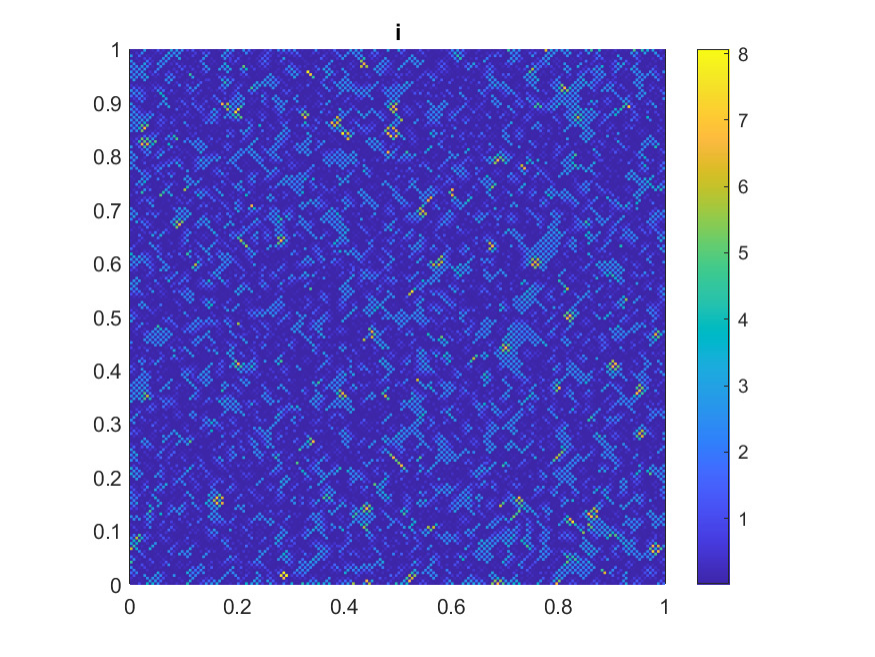}}
    \subfigure[20 days]{
    \includegraphics[scale=0.35]{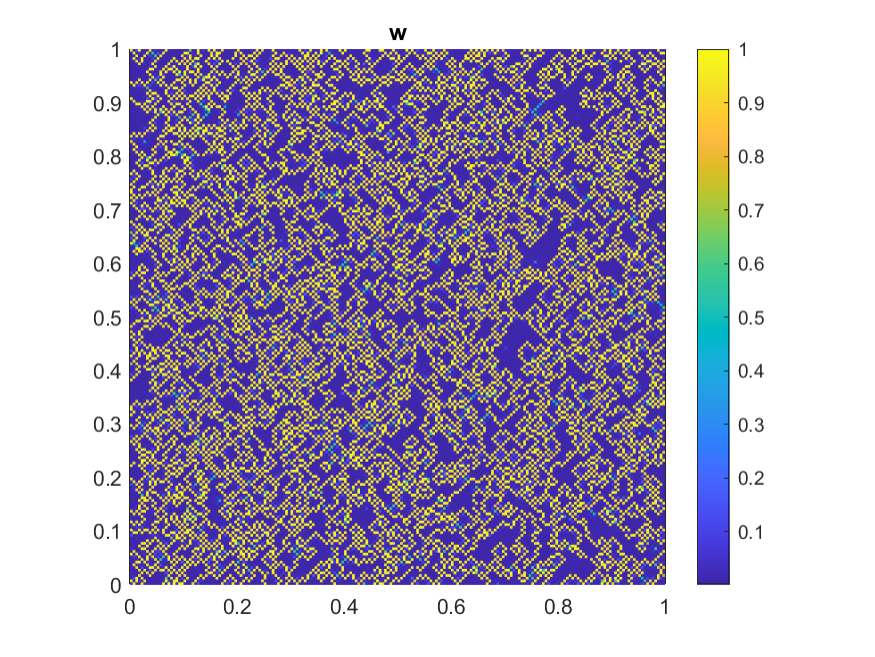}} \\
    \subfigure[30 days]{
    \includegraphics[scale=0.35]{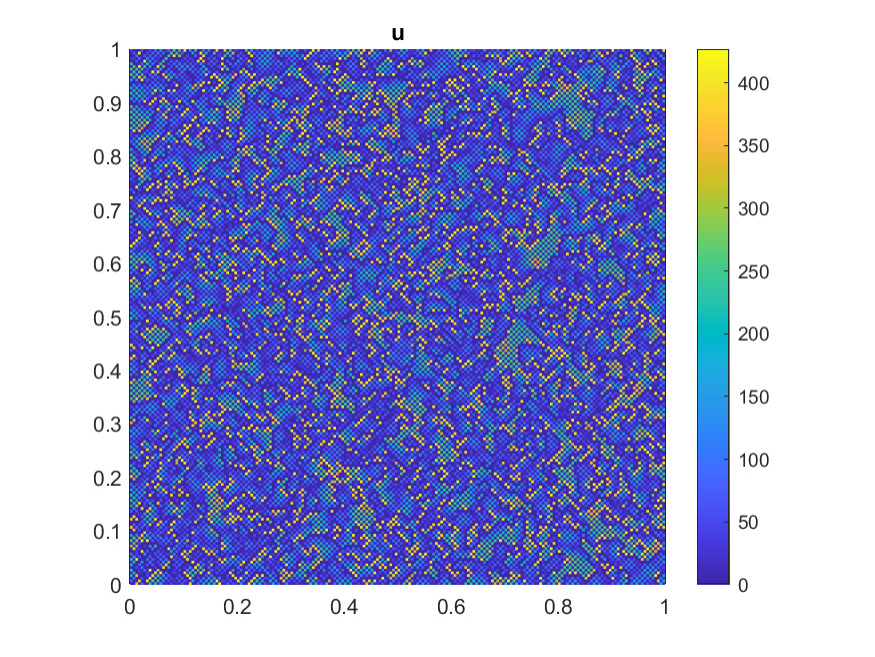}}
    \subfigure[30 days]{
    \includegraphics[scale=0.35]{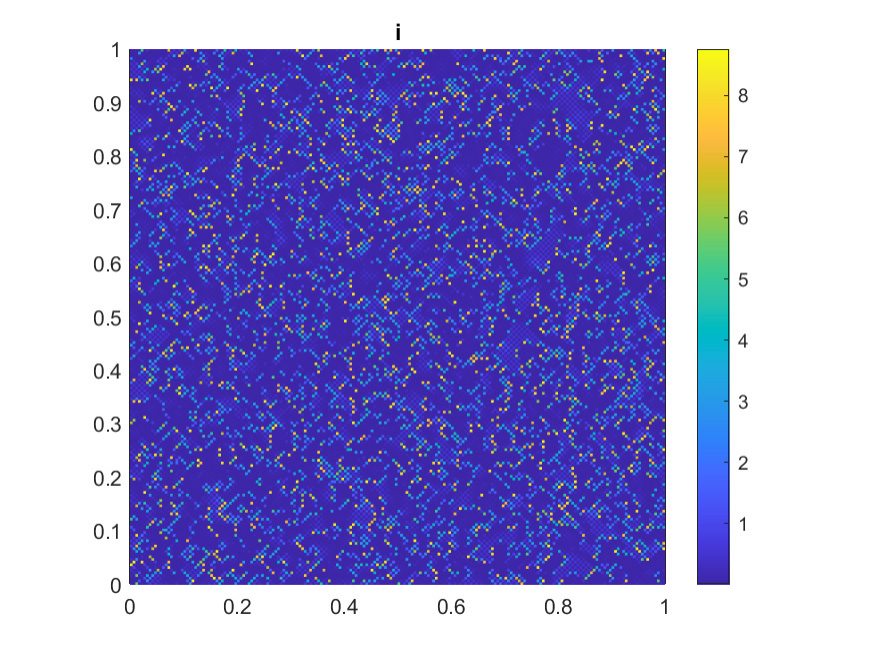}}
    \subfigure[30 days]{
    \includegraphics[scale=0.35]{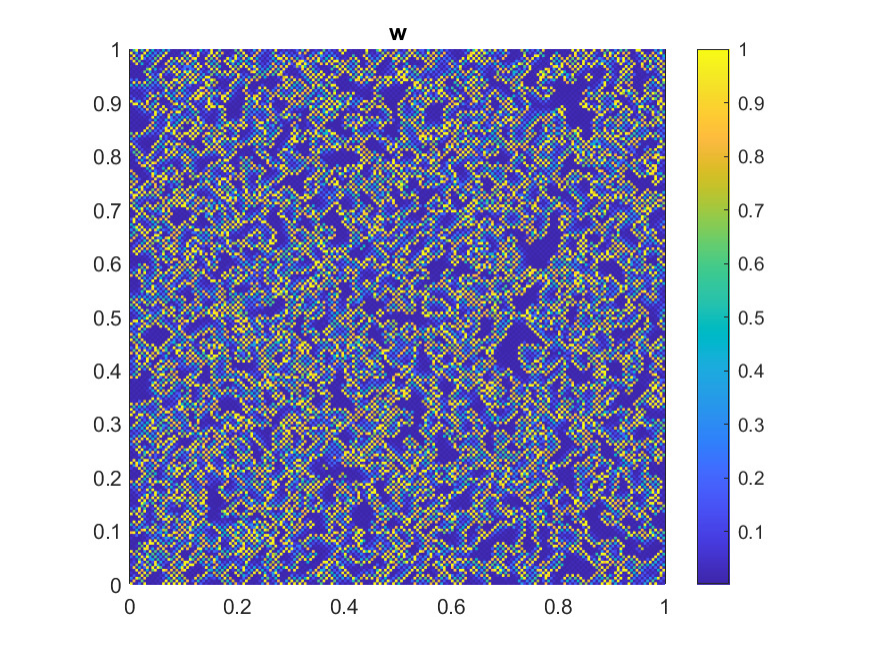}} \\
    \subfigure[100 days]{
    \includegraphics[scale=0.35]{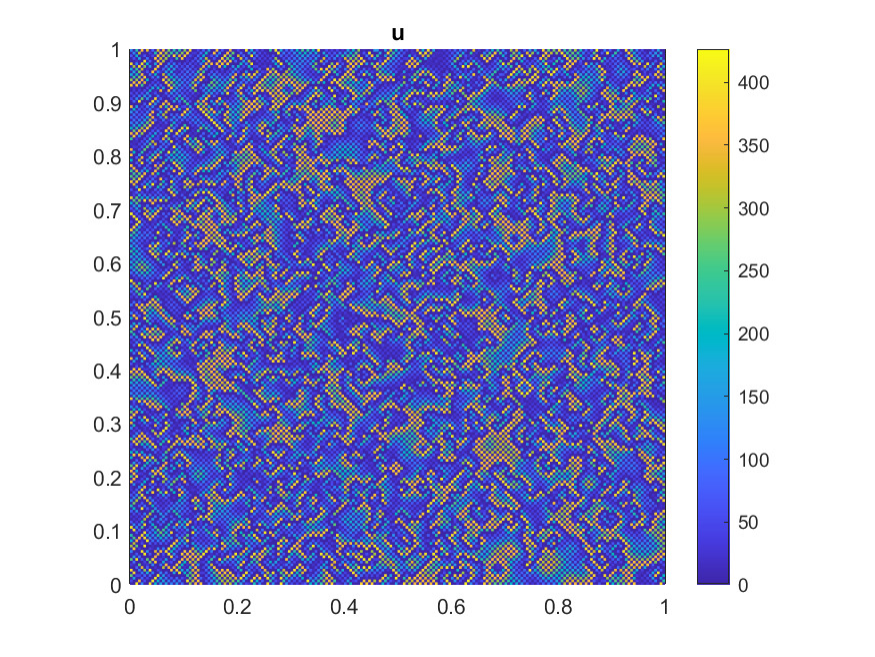}}
    \subfigure[100 days]{
    \includegraphics[scale=0.35]{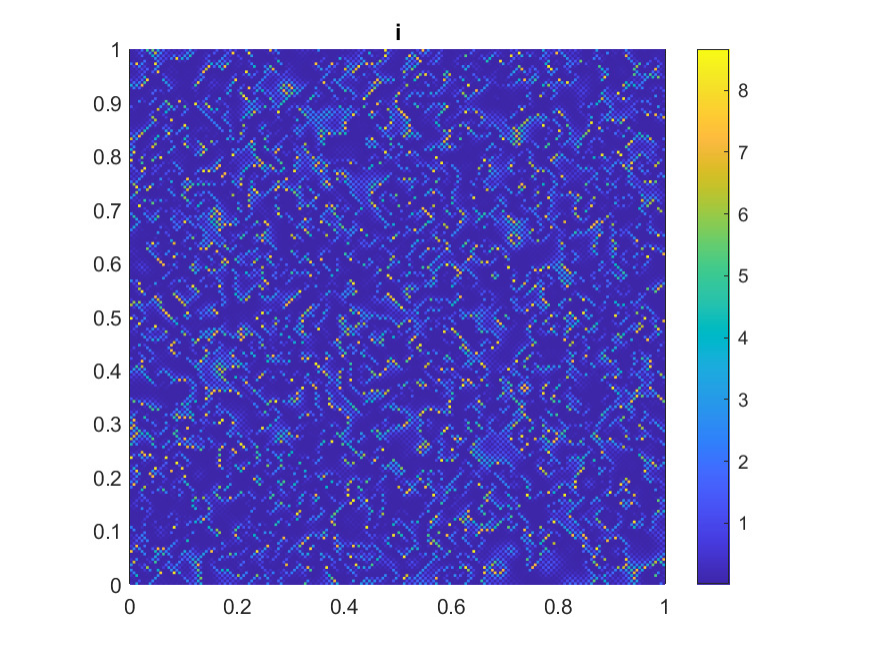}}
    \subfigure[100 days]{
    \includegraphics[scale=0.35]{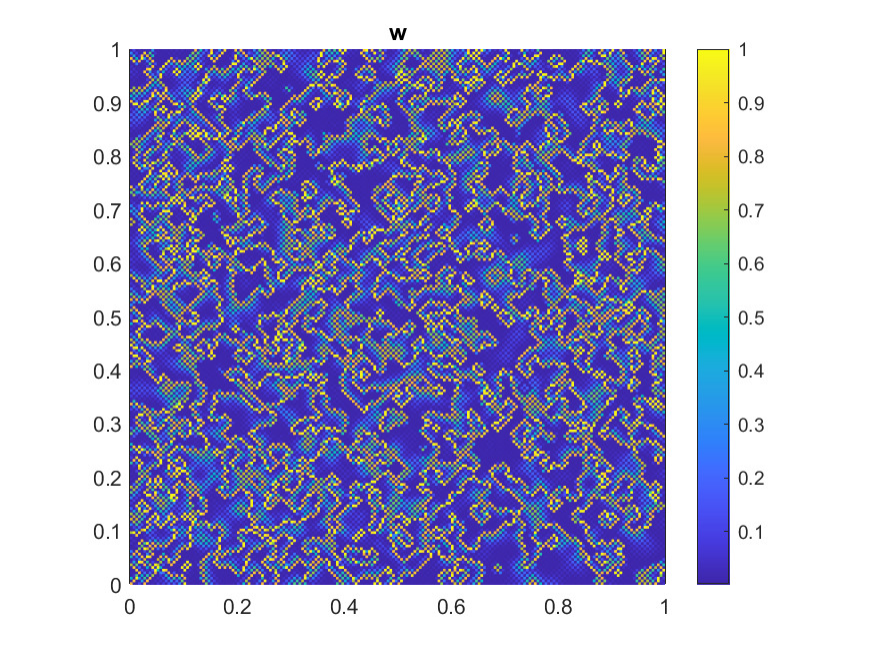}}
    \caption{Spatial distribution of the populations $u,i,w$ plotted choosing as initial conditions a random perturbation of $E^*_3$ and for parameters $S_2 \cup \{\gamma_1=\gamma_2=10^{-9}, \gamma_3=10^{-13},\chi=10^{-1} \}$, after $10$ days in (a)-(c), after $20$ days in (d)-(f), after $30$ days in (g)-(i), after $100$ days in (j)-(l). Discretization space step $\Delta x=\Delta y=0.025$, discretization time step $\Delta t=10^{-3}$.}
    \label{E5tur}
\end{figure}

\begin{figure}[ht!]
    \centering
    \subfigure[]{
    \includegraphics[scale=0.35]{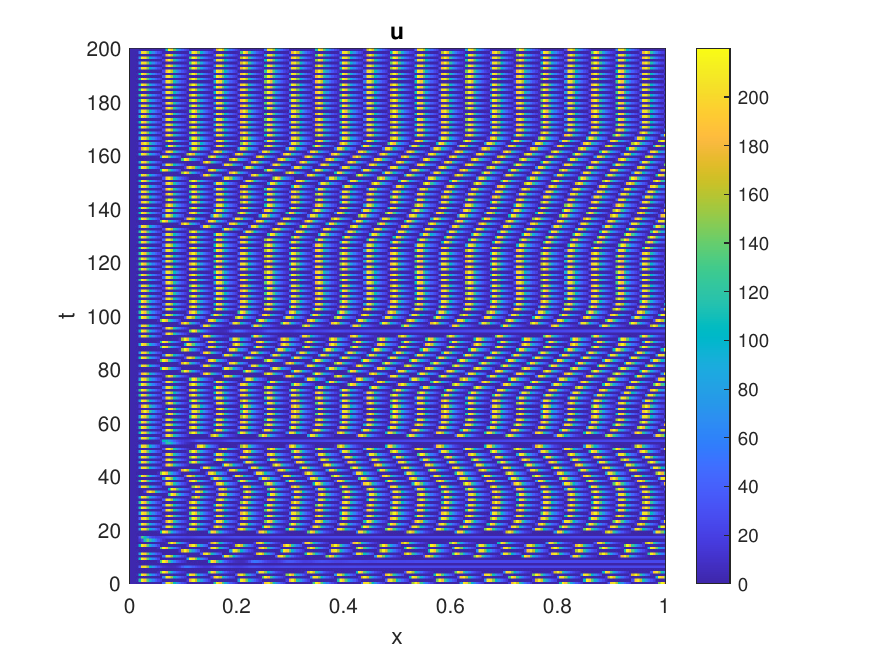}}
    \subfigure[]{
    \includegraphics[scale=0.35]{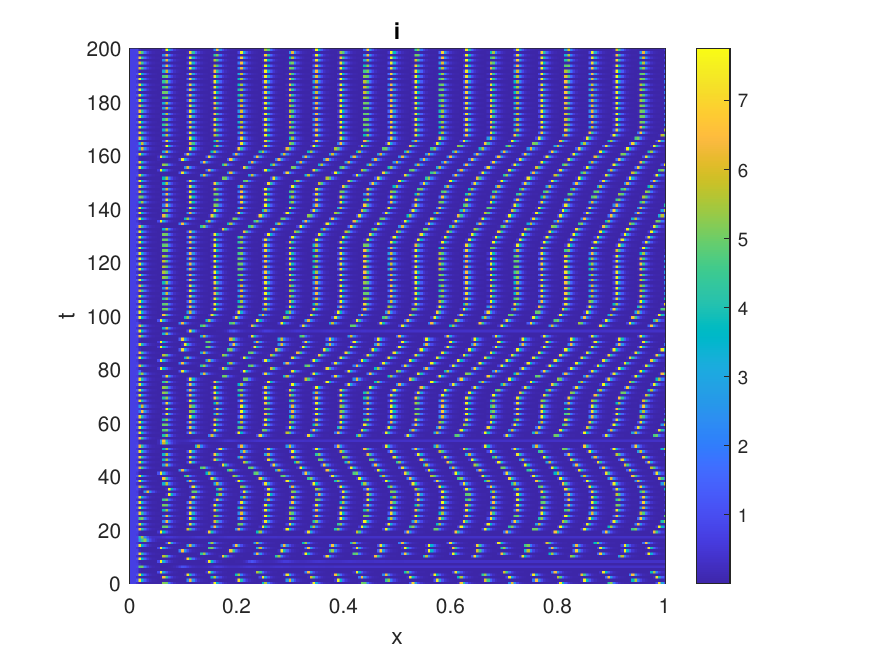}}
    \subfigure[]{
    \includegraphics[scale=0.35]{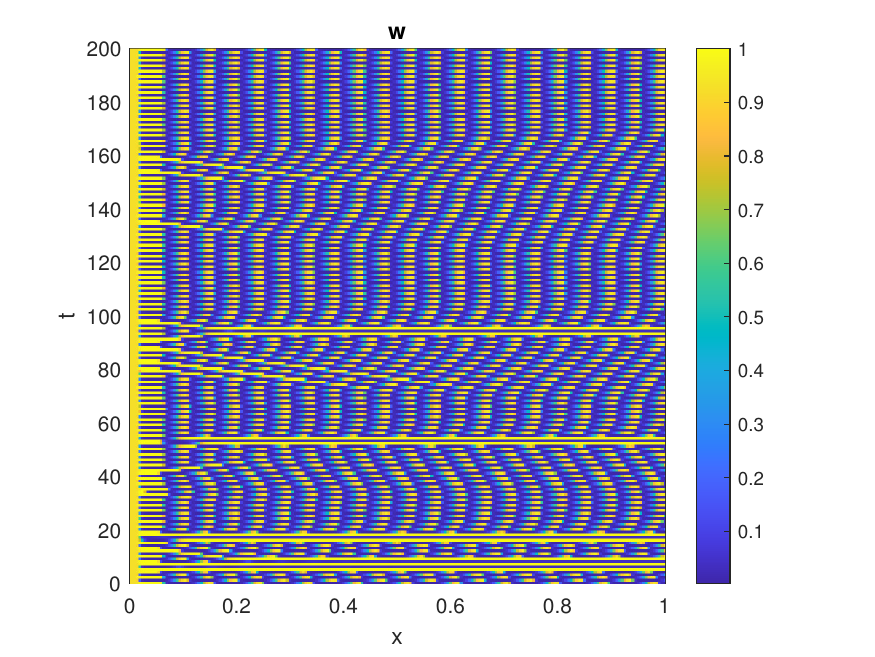}} 
    \caption{Spatio-temporal evolution of the $1D$ numerical solutions $u,i,w$ plotted choosing as initial conditions a random perturbation of $E^*_3$ and for parameters $S_2 \cup \{\gamma_1=\gamma_2=10^{-9}, \gamma_3=10^{-13},\chi=10^{-1} \}$. Discretization space step $\Delta x=0.005$, discretization time step $\Delta t=10^{-3}$.}
    \label{E5_ST}
\end{figure}

\begin{figure}[ht!]
    \centering
    \subfigure[]{
    \includegraphics[scale=0.35]{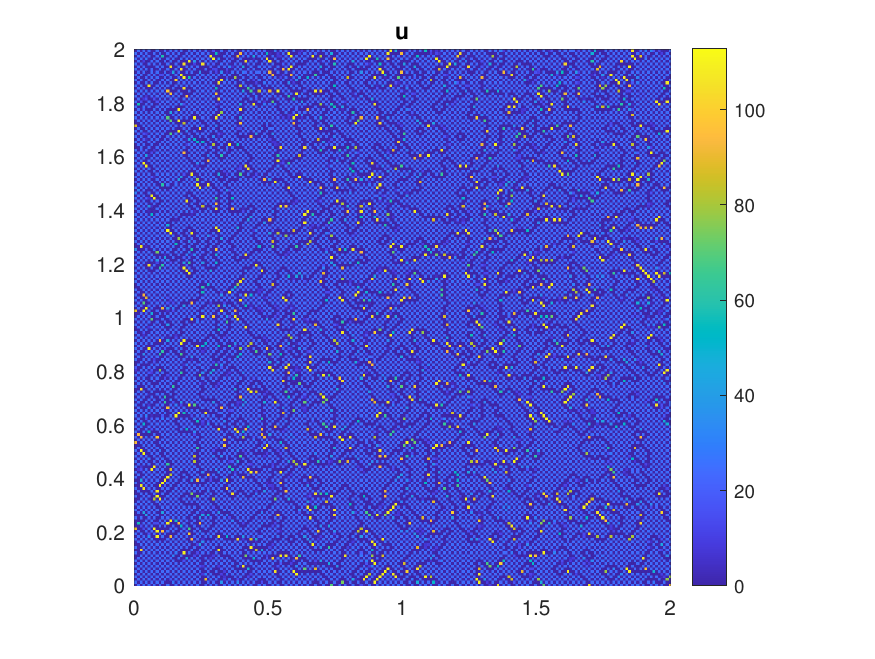}}
    \subfigure[]{
    \includegraphics[scale=0.35]{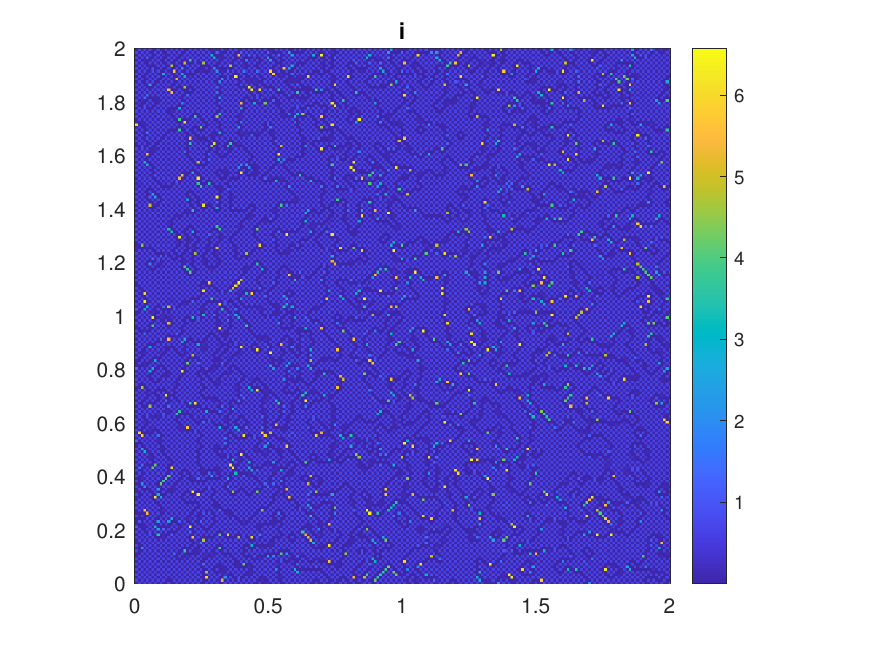}}
    \subfigure[]{
    \includegraphics[scale=0.35]{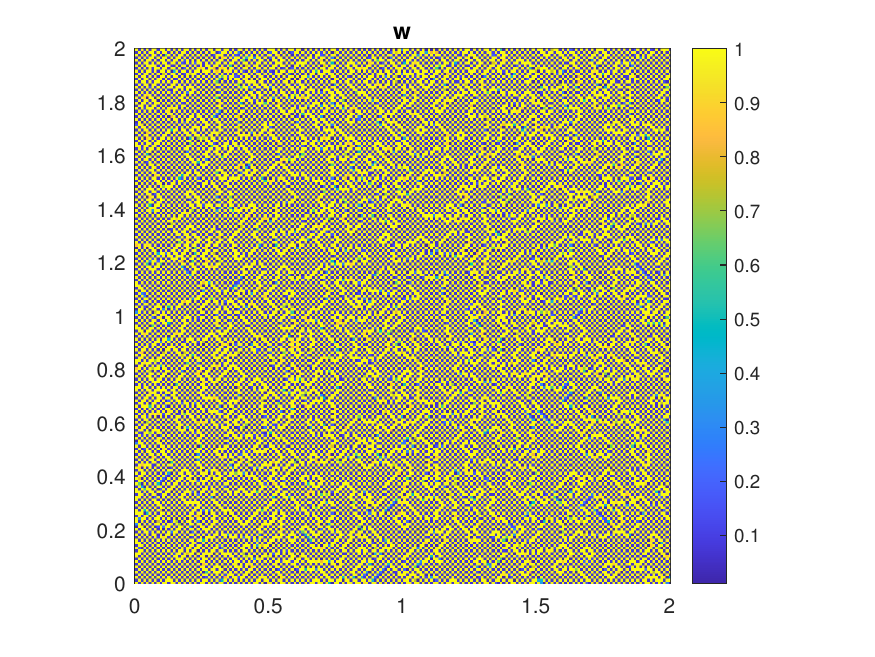}} 
    \caption{Spatial distribution of the populations $u,i,w$ plotted choosing as initial conditions a random perturbation of $E^*_3$ and for parameters $S_2 \cup \{\gamma_1=\gamma_2=10^{-11}, \gamma_3=10^{-14},\chi=10^{-1} \}$, after $100$ days in the spatial domain $[0,2] \times [0,2]$. Discretization space step $\Delta x=\Delta y=0.01$, discretization time step $\Delta t=10^{-4}$.}
    \label{E5_3_tur}
\end{figure}

Finally, let us investigate the results found in the previous Section dedicated to the weakly nonlinear analysis. In the sets of parameters $S_1$ and $S_2$, we plotted $\sigma$ and $L$ defined in \eqref{sigmaL} as function of the bifurcation parameter $\chi$, see Figures \ref{S2_sigmaL_E3}--\ref{S4_sigmaL}. Let us underline that $\sigma$ is positive in the region where Turing patterns are expected (i.e. if $\chi>\chi_c$).  

\begin{figure}[ht!]
    \centering
    \subfigure[]{
    \includegraphics[scale=0.18]{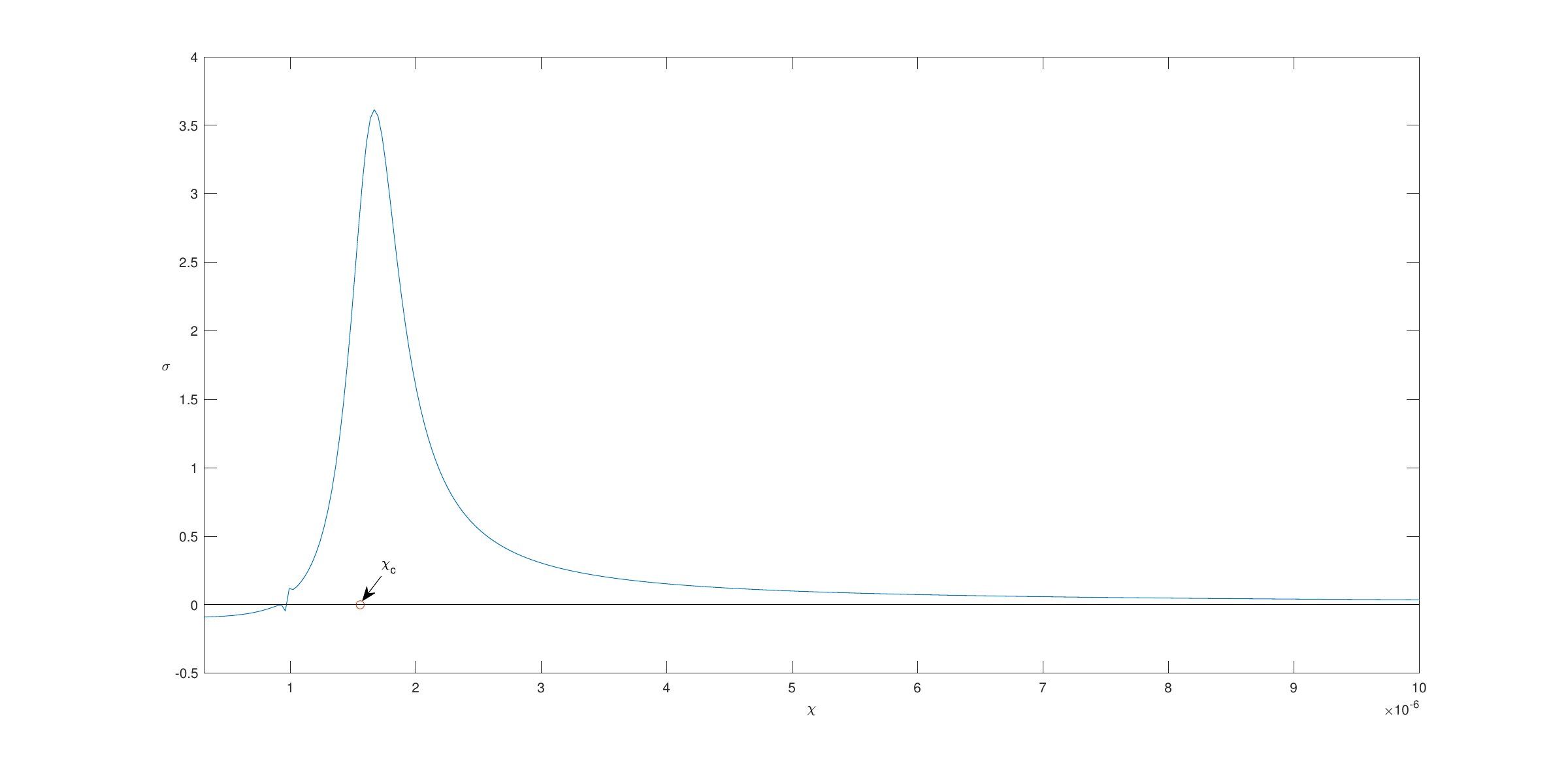}}
    \subfigure[]{
    \includegraphics[scale=0.18]{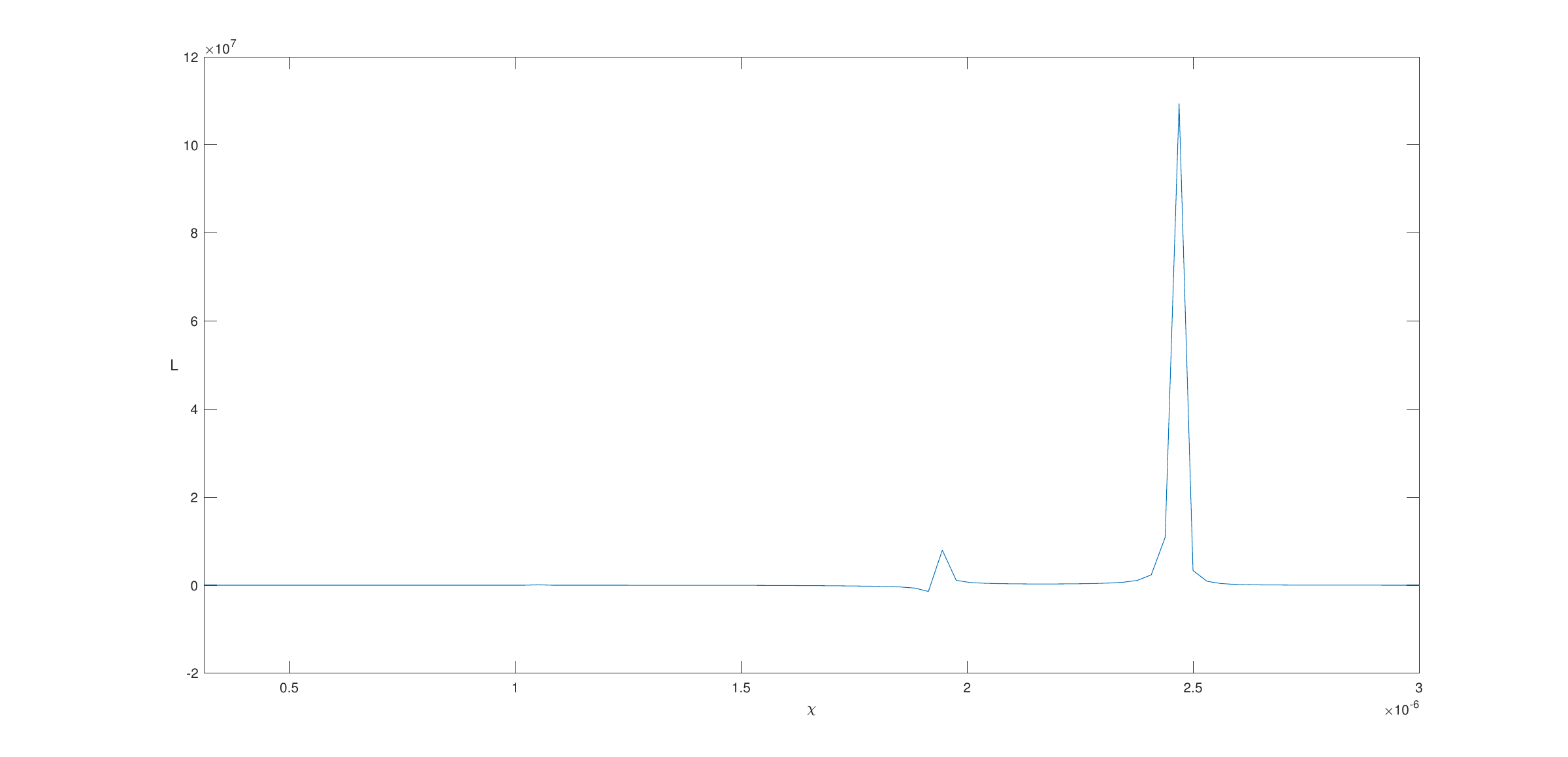}} 
    \caption{Behaviour of $\sigma$ and $L$ defined in \eqref{sigmaL} with respect to $\chi$ for parameters $S_1 \cup \{\gamma_1=\gamma_2=10^{-9}, \gamma_3=10^{-13}\}$, $\chi_c=1.555 \times 10^{-6}$ and for endemic equilibrium $E_2^*$.}
    \label{S2_sigmaL_E3}
\end{figure}

\begin{figure}[ht!]
    \centering
    \subfigure[]{
    \includegraphics[scale=0.18]{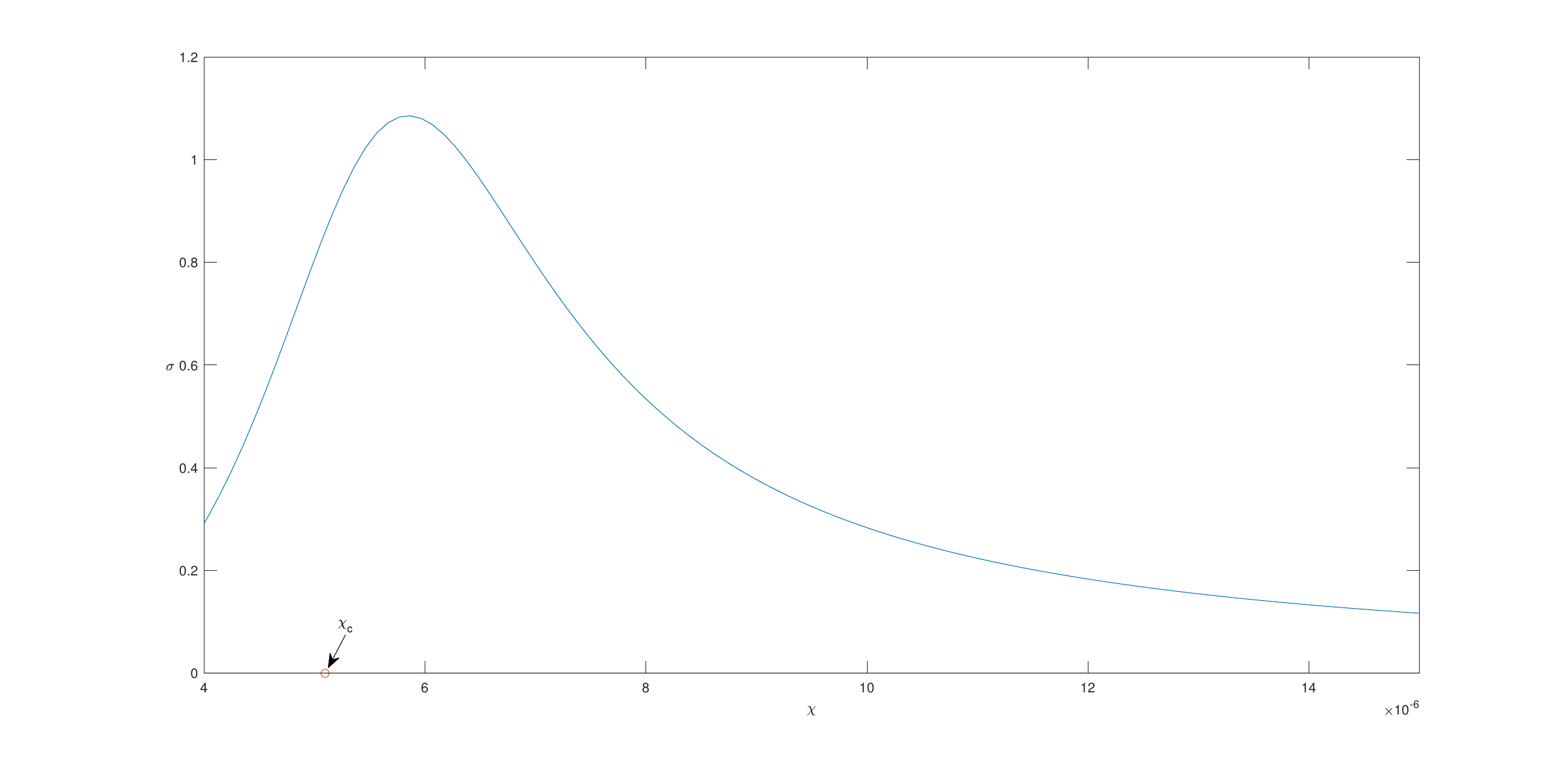}}
    \subfigure[]{
    \includegraphics[scale=0.18]{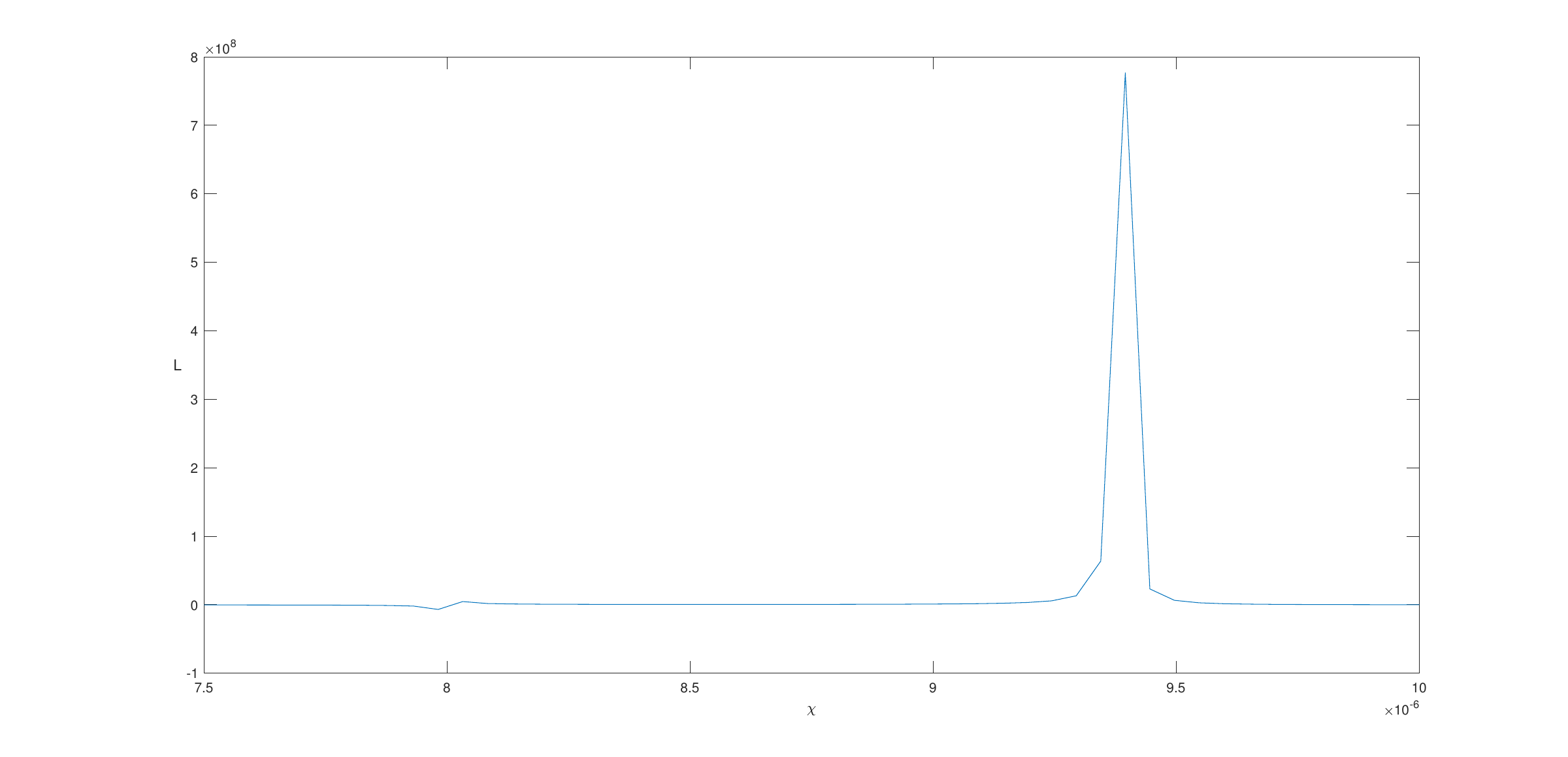}} 
    \caption{Behaviour of $\sigma$ and $L$ defined in \eqref{sigmaL} with respect to $\chi$ for parameters $S_2 \cup \{\gamma_1=\gamma_2=10^{-9}, \gamma_3=10^{-13}\}$, $\chi_c=5.0945 \times 10^{-6}$.}
    \label{S4_sigmaL}
\end{figure}

\section{Conclusion}
In this paper, we analysed a model for \emph{Mycobacterium tuberculosis} infection, taking into account a random movement of the populations, hence, both self-diffusion and chemotaxis are considered. Via linear instability analysis of the endemic equilibrium $\bar{E}$ (when it exists), we proved that in presence of both self-diffusion and chemotaxis, Turing instability can occur for $\bar{E}$, therefore we found sufficient conditions guaranteeing linear stability of $\bar{E}$ in the absence of diffusion and linear instability of $\bar{E}$ in presence of diffusion. Moreover, we perform a weakly nonlinear analysis and the solutions of the Landau amplitude equations have been studied. Finally, via numerical simulations, the results we found have been analysed and the interaction between macrophages and bacteria populations has been analysed.

\section*{Acknowledgments}
This paper has been performed under the auspices of the GNFM of INdAM.
R. De Luca and F. Capone acknowledge the support of National Recovery and Resilience Plan (NRRP) funded by the European Union - NextGenerationEU - Project Title ``Mathematical Modeling of Biodiversity in the Mediterranean sea: from bacteria to predators, from meadows to currents" - project code $P202254HT8$ (CUP $B53D23027760001$) and the support of grant no. MUR-PRIN PNNR $2022$ - Project Title ``Modelling complex biOlogical systeMs for biofuEl productioN and sTorAge: mathematics meets green industry" - project code $202248TY47$ (CUP $E53D23005430006$).

\bibliographystyle{unsrt}
\bibliography{references}

\begin{thebibliography}{10}

\bibitem{WHO2023}
World~Health Organization et~al.
\newblock Global tuberculosis report 2023.
\newblock In {\em Global tuberculosis report 2023}. 2023.

\bibitem{Du2017}
Y.~Du, J.~Wu, and J.~M. Heffernan.
\newblock A simple in-host model for {M}ycobacterium tuberculosis that captures
  all infection outcomes.
\newblock {\em Mathematical Population Studies}, 24(1):37--63, 2017.

\bibitem{ZFH2020}
W.~Zhang, F.~Frascoli, and J.M. Heffernan.
\newblock Analysis of solutions and disease progressions for a within-host
  tuberculosis model.
\newblock {\em Mathematics in Applied Sciences and Engineering}, 1(1):39--49,
  2020.

\bibitem{ZFH2021}
W.~Zhang, L.~Ellingson, F.~Frascoli, and J.~Heffernan.
\newblock An investigation of tuberculosis progression revealing the role of
  macrophages apoptosis via sensitivity and bifurcation analysis.
\newblock {\em Journal of mathematical biology}, 83(3):31, 2021.

\bibitem{zhang2020}
W.~Zhang.
\newblock Analysis of an in-host tuberculosis model for disease control.
\newblock {\em Applied Mathematics Letters}, 99:105983, 2020.

\bibitem{natalini2010}
F.~Clarelli and R.~Natalini.
\newblock A pressure model of immune response to {M}ycobacterium tuberculosis
  infection in several space dimensions.
\newblock {\em Math Biosci Eng}, 7(2):277--300, 2010.

\bibitem{T1952}
A.M. Turing.
\newblock The {C}hemical {B}asis of {M}orphogenesis.
\newblock {\em Philosophical Transactions of the Royal Society of London.
  Series B, Biological Sciences}, 237(641):37--72, 1952.

\bibitem{GLS2021}
V.~Giunta, M.C. Lombardo, and M.~Sammartino.
\newblock Pattern formation and transition to chaos in a chemotaxis model of
  acute inflammation.
\newblock {\em SIAM Journal on Applied Dynamical Systems}, 20(4):1844--1881,
  2021.

\bibitem{CDLFLM2023}
F.~Capone, R.~De~Luca, L.~Fiorentino, V.~Luongo, and G.~Massa.
\newblock Turing instability for a {L}eslie--{G}ower model.
\newblock {\em Ricerche di Matematica}, pages 1--18, 2023.

\bibitem{CDL2017}
F.~Capone and R.~De~Luca.
\newblock On the nonlinear dynamics of an ecoepidemic reaction--diffusion
  model.
\newblock {\em International Journal of Non-Linear Mechanics}, 95:307--314,
  2017.

\bibitem{CDLT2018}
F.~Capone, R.~De~Luca, and I.~Torcicollo.
\newblock Influence of diffusion on the stability of a full brusselator model.
\newblock {\em Rendiconti Lincei}, 29(4):661--678, 2018.

\bibitem{GLSS2013}
G.~Gambino, M.C. Lombardo, M.~Sammartino, and V.~Sciacca.
\newblock Turing pattern formation in the {B}russelator system with nonlinear
  diffusion.
\newblock {\em Physical Review E}, 88(4):042925, 2013.

\bibitem{GammackEtAl2004}
D.~Gammack, C.~R. Doering, and D.~E. Kirschner.
\newblock Macrophage response to {M}ycobacterium tuberculosis infection.
\newblock {\em Journal of mathematical biology}, 48:218--242, 2004.

\bibitem{symon1972plasma}
D.~N.~K. Symon, I.~C. McKay, and P.~C. Wilkinson.
\newblock Plasma-dependent chemotaxis of macrophages towards {M}ycobacterium
  tuberculosis and other organisms.
\newblock {\em Immunology}, 22(2):267, 1972.

\bibitem{Murray2_2002}
J.~D Murray.
\newblock {\em Mathematical biology {II}: {S}patial models and biomedical
  applications}, volume~3.
\newblock Springer New York, 2002.

\bibitem{BGMS2024}
M.~Bisi, M.~Groppi, G.~Martal{\`o}, and C.~Soresina.
\newblock A chemotaxis reaction--diffusion model for {M}ultiple {S}clerosis
  with {A}llee effect.
\newblock {\em Ricerche di Matematica}, 73(Suppl 1):29--46, 2024.

\bibitem{TLS2014}
E.~Tulumello, M.C. Lombardo, and M.~Sammartino.
\newblock Cross-diffusion driven instability in a predator-prey system with
  cross-diffusion.
\newblock {\em Acta Applicandae Mathematicae}, 132(1):621--633, 2014.

\bibitem{LBBGPS2017}
M.C. Lombardo, R.~Barresi, E.~Bilotta, F.~Gargano, P.~Pantano, and
  M.~Sammartino.
\newblock Demyelination patterns in a mathematical model of multiple sclerosis.
\newblock {\em Journal of mathematical biology}, 75:373--417, 2017.

\bibitem{WK2001}
J.~E. Wigginton and D.~Kirschner.
\newblock A model to predict cell-mediated immune regulatory mechanisms during
  human infection with {M}ycobacterium tuberculosis.
\newblock {\em The Journal of Immunology}, 166(3):1951--1967, 2001.

\bibitem{MK2004}
S.~Marino and D.~E. Kirschner.
\newblock The human immune response to {M}ycobacterium tuberculosis in lung and
  lymph node.
\newblock {\em Journal of theoretical biology}, 227(4):463--486, 2004.

\bibitem{GammackEtAl2005}
D.~Gammack, S.~Ganguli, S.~Marino, J.~Segovia-Juarez, and D.~E. Kirschner.
\newblock Understanding the immune response in tuberculosis using different
  mathematical models and biological scales.
\newblock {\em Multiscale Modeling \& Simulation}, 3(2):312--345, 2005.

\bibitem{SudEtAl2006}
D.~Sud, C.~Bigbee, J.~L Flynn, and D.~E. Kirschner.
\newblock Contribution of cd8+ t cells to control of {M}ycobacterium
  tuberculosis infection.
\newblock {\em The Journal of Immunology}, 176(7):4296--4314, 2006.

\bibitem{Murray1_2002}
J.~D. Murray.
\newblock {\em Mathematical biology {I}: {A}n introduction}, volume~3.
\newblock Springer, 2002.

\bibitem{Merkin1997}
D.R. Merkin.
\newblock Stability of linear autonomous systems.
\newblock In {\em Introduction to the theory of stability}, pages 133--158.
  Springer, 1997.

\bibitem{Garvie2007}
M.R. Garvie.
\newblock Finite-difference schemes for reaction--diffusion equations modeling
  predator--prey interactions in {M}atlab.
\newblock {\em Bulletin of mathematical biology}, 69:931--956, 2007.

\end{thebibliography}

\end{document}